\documentclass[12pt]{article}
\usepackage{amsmath,amssymb,bbm,epsfig,wrapfig,xcolor}
\usepackage[left=1in, right=1in, top=1in, bottom=1in]{geometry}
\usepackage{lmodern}
\usepackage{slantsc}
\usepackage{ytableau}
\usepackage{sectsty}
\subsubsectionfont{\centering\fontsize{12}{15}\selectfont}
\subsectionfont{\centering\fontsize{12}{15}\selectfont}
\sectionfont{\centering\fontsize{14}{15}\selectfont}
\usepackage[english]{babel}
\usepackage[utf8]{inputenc}
\usepackage[affil-it]{authblk}
\usepackage{verbatim}
\usepackage{enumitem}
\usepackage{xfrac}
\usepackage{tikz}
\usepackage{amsthm}
\usepackage{amsfonts}
\usepackage{mathtools}
\usepackage[colorlinks]{hyperref}
\usepackage{graphicx}
\usepackage{subcaption}
\usepackage{caption}
\usepackage{appendix}
\usepackage{cleveref}
\usepackage[colorinlistoftodos]{todonotes}
\usepackage{multirow}
\usepackage{url}
\newcommand{\newc}{\newcommand}

\newc{\beq}{\begin{equation}}
\newc{\eeq}{\end{equation}}
\newc{\kt}{\rangle}
\newc{\br}{\langle}
\newc{\Tr}{\mbox{{\rm Tr}}}
\newc{\intpar}{\mbox{{\rm par}}}

\DeclareMathOperator{\sgn}{sgn}
\newtheorem{theorem}{Theorem}
\newtheorem{corollary}{Corollary}[theorem]
\newtheorem*{remark}{Remark}

\newtheorem{lemma}{Lemma}[theorem]
\newtheorem{proposition}{Proposition}[theorem]
\numberwithin{theorem}{section}
\numberwithin{corollary}{section}
\numberwithin{lemma}{section}
\numberwithin{proposition}{section}
\numberwithin{figure}{section}			

\numberwithin{equation}{section}		
\renewcommand\theequation{\arabic{section}.\arabic{equation}}	
\usepackage{array}
\hypersetup{
    colorlinks=true,
    linkcolor=blue,
    filecolor=magenta,      
    citecolor=blue,
    urlcolor = blue
}
\newcommand*\pFqskip{8mu}
\catcode`,\active
\newcommand*\pFq{\begingroup
        \catcode`\,\active
        \def ,{\mskip\pFqskip\relax}%
        \dopFq
}
\catcode`\,12
\def\dopFq#1#2#3#4#5{%
        {}_{#1}F_{#2}\biggl(\genfrac..{0pt}{}{#3}{#4};#5\biggr)%
        \endgroup
}

\begin{document}

\title{Symmetric Function Theory and Unitary Invariant Ensembles\footnote{Revised version: 12 August 2021}}

\author[1]{ Bhargavi Jonnadula}
\affil[1]{\small School of Mathematics, University of Bristol, Fry Building,
Bristol,
BS8 1UG, UK}
\author[2]{Jonathan P. Keating}
\affil[2]{\small Mathematical Institute, University of Oxford, Andrew Wiles Building, Oxford, OX2 6GG, UK}
\author[1]{Francesco Mezzadri}

\date{}


\maketitle

\begin{abstract}
Representation theory and the theory of symmetric functions have played a central role in Random Matrix Theory in the computation of quantities such as joint moments of traces and joint moments of characteristic polynomials of matrices drawn from the Circular Unitary Ensemble and other Circular Ensembles related to the classical compact groups.  The reason is that they enable the derivation of exact formulae, which then provide a route to calculating the large-matrix asymptotics of these quantities.
We develop a parallel theory for the Gaussian Unitary Ensemble of random matrices, and other related unitary invariant matrix ensembles.  This allows us to write down exact formulae in these cases for the joint moments of the traces and the joint moments of the characteristic polynomials in terms of appropriately defined symmetric functions.  As an example of an application, for the joint moments of the traces we derive explicit asymptotic formulae for the rate of convergence of the moments of polynomial functions of GUE matrices to those of a standard normal distribution when the matrix size tends to infinity.
\end{abstract}

\section{Introduction}


Many important quantities in Random Matrix Theory, such as joint moments of traces and joint moments of characteristic polynomials, can be calculated exactly for matrices drawn from the Circular Unitary Ensemble and the other Circular Ensembles related to the classical compact groups using representation theory and the theory of symmetric polynomials.  In the case of joint moments of the traces, this approach has proved highly successful, as in, for example, the work of Diaconis and Shahshahani \cite{Diaconis1994}.  Similarly, the joint moments of characteristic polynomials were calculated exactly in terms of Schur polynomials by Bump and Gamburd \cite{Bump2006}, leading to expressions equivalent to those obtained using the Selberg integral and related techniques \cite{Baker1997finite, Keating2000random, Keating2000, Conrey2003}.  Lee and Oh~\cite{Lee2020} extended the work of Bump and Gamburd \cite{Bump2006} and computed the correlation functions of characteristic polynomials in the Sato-Tato groups as sums of characters of irreducible characters of the symplectic group $Sp(N)$. Our aim here is to develop a parallel theory for the classical unitary invariant Hermitian ensembles of random matrices, in particular for the Gaussian (GUE), Laguerre (LUE), and Jacobi ensembles (JUE).

Characteristic polynomials and their asymptotics have been well studied for Hermitian matrices using orthogonal polynomials, super-symmetric techniques, Selberg and Itzykson-Zuber integrals, see, for example, \cite{BH00,Borodin2006,Fyodorov2002,Fyodorov2002chiralGUE,Fyodorov2003,Baik2003}.  
Other properties including universality \cite{Strahov2003,Breuer2012}, and ensembles with external sources \cite{Fyodorov2018,Forrester2013} have also been considered. Here we give a symmetric-function-theoretic approach similar to that established by Bump and Gamburd \cite{Bump2006}, using generalised Schur polynomials \cite{Sergeev2014} or multivariate orthogonal polynomials \cite{Baker1997,Baker1997calogero} to compute correlation functions of characteristic polynomials for $\beta=2$ ensembles. 

Diaconis and Shashahani \cite{Diaconis1994} used group-theoretic arguments and symmetric functions to calculate joint moments of traces of matrices for classical compact groups. Here, using multivariate orthogonal polynomials, we develop a similar approach to calculate joint moments of traces for Hermitian ensembles, leading to closed form expressions using combinatorial and symmetric-function-theoretic methods. 

Moments of Hermitian ensembles and their correlators have recently received considerable attention. Cunden \emph{et al.}~\cite{Cunden2019} showed that as a function of their order, the moments are hypergeometric orthogonal polynomials.  Cunden, Dahlqvist and O'Connell~\cite{Cunden2021} showed that the cumulants of the Laguerre ensemble admit an asymptotic expansion in inverse powers of $N$ of whose coefficients are the Hurwitz numbers.  Dubrovin and Yang~\cite{Dubrovin2017} computed the cumulant generating function for the GUE, while Gisonni, Grava and Ruzza calculated the generating function of the cumulants of the LUE in \cite{Gisonni2020} and the JUE in \cite{Gisonni2020jacobi}.

If $M$ is drawn at random from the classical compact groups $ U(N), O(N), Sp(N)$ equipped with Haar measure, then $\Tr M^k$, $k\in\mathbb{N}$, converges to a complex normal random variable as $N \to \infty$.  Johansson~\cite{Johansson1998} was the first to prove a central limit theorem when $M$ belongs to an ensemble of Hermitian matrices invariant under unitary conjugation.  In this case the analogue of  $\Tr M^k$ is played by $\Tr \, T_k(M)$, where $T_k$ is the Chebyshev polynomial of the first kind; see also \cite{Borot2013,Witte2014,Shcherbina2007,Forrester2017large,Forrester2017,Simm2017,Berezin2021,Pastur2006,Nakano2018,Lambert2019,Bekerman2018,Kopel2015} and references therein. 
As an example of an application of the general approach we take here, we apply our results to establish explicit asymptotic formulae for the rate of convergence of the moments and cumulants of Chebyshev-polynomial functions of GUE matrices to those of a standard normal distribution when the matrix size tends to infinity. In a companion article~\cite{Jonnadula2021}, we use the techniques developed in this paper to investigate the moments of the characteristic polynomials in the GUE,  uncovering structure that had been overlooked in previous studies.

The theory of symmetric functions has been applied to orthogonal polynomials also outside 
the context of random matrices. In a spirit similar to that in this paper, generalizations of classical combinatorial identities like the Cauchy formula have played a fundamental role.   The Schur functions arise naturally  in the representation theory of the Heisenberg algebra.  Lam~\cite{Lam2006} showed that by choosing a particular representation,  the Schur polynomials are replaced by a new class of symmetric functions that obey
Pieri and Cauchy-like identities.  In turn, these can be applied to the theory of Hall-Littlewood  and Macdonald polynomials. More recently, Borodin~\cite{Borodin2017} studied a one-parameter family of rational symmetric functions that generalize the Hall-Littlewood polynomials.  He also showed that such symmetric functions satisfy Pieri and Cauchy-like identities.  



This paper is structured as follows. In Sec.~\ref{sec:statements_and_results} we introduce our main results. The preliminaries, multivariate orthogonal polynomials and their properties are discussed in Sec.~\ref{sec:background}. The correlation functions of characteristic polynomials are calculated in Sec.~\ref{sec:correlation_functions}. We discuss the change of basis among different symmetric functions in Sec.~\ref{sec:moments} and prove there the results for the moments of characteristic polynomials and the joint moments of traces for different ensembles. Finally, in Sec.~\ref{sec:fluctuations}, by way of an example, we apply the results to derive explicit asymptotic formulae for the rate of convergence of the moments and cumulants of Chebyshev-polynomial functions of GUE matrices to those of a standard normal distribution when the matrix size tends to infinity.

\section{Statements and results}\label{sec:statements_and_results}
For the classical compact groups, Schur polynomials and their generalisations are the characters of $U(N)$, $O(N)$ and $Sp(N)$.  In this context they have been used extensively to calculate correlation functions of characteristic polynomials, joint moments of the traces, see \cite{Diaconis1994,Bump2006}.  Although group theoretic tools are not available for the set of Hermitian matrices, multivariate orthogonal polynomials play the role of Schur functions for the GUE, LUE and JUE and can be used to study fundamental quantities like moments of traces and characteristic polynomials.

For a partition $\mu$, let $\varPhi_\mu$ be the multivariate symmetric polynomials with leading coefficient equal to 1 that obey the orthogonality relation
\beq\label{eq:mul_poly}
\int \varPhi_\mu(x_1,\dots,x_N)\varPhi_\nu(x_1,\dots,x_N)\prod_{1\leq i<j\leq N}(x_i-x_j)^2\prod_{j=1}^Nw(x_j)\,dx_j = \delta_{\mu\nu}C_\mu
\eeq
for a weight function $w$.
Here the lengths of the partitions $\mu$ and $\nu$ are less than or equal to the number of variables $N$, and $C_\mu$ is a constant which depends on $N$. We prove the following lemma, which is a generalization of the dual Cauchy identity. 
\begin{lemma}\label{lemma:dualcauchy_to_mulmonicpoly}
Let $\varPhi_\mu$ be multivariate polynomials given in \eqref{eq:mul_poly}. Let $p,q\in\mathbb{N}$ and for $\lambda\subseteq (q^p)\equiv (\underbrace{q,\dots,q}_{p})$ let $\tilde{\lambda}=(p-\lambda_q^\prime,\dots,p-\lambda_1^\prime)$. Then
\beq\label{eq:new_identity}
\prod_{i=1}^p\prod_{j=1}^q(t_i-x_j) = \sum_{\lambda\subseteq (q^p)} (-1)^{|\tilde{\lambda}|}\varPhi_\lambda(t_1,\dots,t_p)\varPhi_{\tilde{\lambda}}(x_1,\dots ,x_q).
\eeq
\end{lemma}
Here partition $\lambda=(\lambda_1,\dots,\lambda_l)$ such that $\lambda_1\geq\dots\geq\lambda_l$ is a sub-partition of the partition $(q^p)$, denoted by $\lambda\subseteq (q^p)$. (See Sec.~\ref{revsymfun}). This lemma appears in \cite[p.625]{Forrester2010} for the Jacobi multivariate polynomials for arbitrary $\beta$. 
Here we present a different proof for $\beta=2$, which holds for the Hermite and Laguerre polynomials, too. A key difference in our approach is that we have closed-form expressions for multivariate polynomials as determinants of univariate classical orthogonal polynomials, while in the previous literature their construction was based on recurrence relations.   This means that in this paper formula~\eqref{eq:new_identity} becomes a powerful tool and plays a role analogous to that of the  classical dual Cauchy identity for $U(N)$. It is worth noting that the identity in Lemma~\ref{lemma:dualcauchy_to_mulmonicpoly} is independent of the weight $w$. This is a consequence of the row and column operations on determinants (see Eqs.~\eqref{eq:dummy_1} and~\eqref{eq:dummy_2}) that are central to the proof. It is the analogue of the well known
equality for Vandermonde determinants,
\beq
\prod_{1 \le i < j \le N}(x_i - x_j) = \det\left[ x^{N-j}_i\right]_{i,j=1,\dotsc,N} =\det\left[ \varphi_{N-j}(x_i)\right]_{i,j=1,\dotsc,N}, 
\eeq
where $\varphi_i$, $i=1,2\dotsc, $ is any sequence of monic orthogonal polynomials.


We focus in particular on when $w(x)$ in \eqref{eq:mul_poly} is a Gaussian, Laguerre and Jacobi weight:
\beq\label{eq:weights}
w(x) = 
\begin{cases}
e^{-\frac{x^2}{2}}, \qquad\qquad\,\, x\in \mathbb{R}, \qquad\qquad\qquad\quad\quad\,\,\, \text{Gaussian},\\
x^\gamma e^{-x},  \qquad\quad\,\,\,\, x\in \mathbb{R}_+,\qquad \gamma>-1, \qquad\,\,\text{Laguerre},\\
x^{\gamma_1}(1-x)^{\gamma_2}, \quad x\in[0,1], \quad \gamma_1,\gamma_2>-1, \quad \text{Jacobi}.
\end{cases}
\eeq
The classical polynomials orthogonal with respect to these weights satisfy 
\begin{subequations}
\label{eq:classical_hlj_ortho}
\begin{gather}
\label{eq:classical_hlj_ortho1}
\int_{\mathbb{R}}H_j(x)H_k(x)e^{-\frac{x^2}{2}}\,dx = \sqrt{2\pi}j!\delta_{jk},\\
\label{eq:classical_hlj_ortho2}
 \quad\int_{\mathbb{R}_{+}}L^{(\gamma)}_mL^{(\gamma)}_nx^\gamma e^{-x}\, dx 
= \frac{\Gamma(n+\gamma +1)}{\Gamma(n+1)}\delta_{nm},\\
\int_0^1 J^{(\gamma_1,\gamma_2)}_n(x)J^{(\gamma_1,\gamma_2)}_m(x)x^{\gamma_1}(1-x)^{\gamma_2}\,dx \nonumber\\
\label{eq:classical_hlj_ortho3}
= \frac{1}{(2n+\gamma_1+\gamma_2+1)}\frac{\Gamma(n+\gamma_1+1)\Gamma(n+\gamma_2+1)}{n!\Gamma(n+\gamma_1+\gamma_2+1)}\delta_{mn}.
\end{gather}
\end{subequations}
The identity in \eqref{eq:new_identity} gives a compact way to calculate the correlation functions and moments of characteristic polynomials of unitary ensembles using symmetric functions.  The results are as follows.
\begin{theorem}\label{thm:correlations_charpoly}
Let $M$ be an $N\times N$ GUE, LUE or JUE matrix and $t_1,\dots, t_p\in\mathbb{C}$. Then,
\beq
\begin{split}
\textit{(a)}\quad\mathbb{E}^{(H)}_N[\prod_{j=1}^p\det(t_j - M)] &=\mathcal{H}_{(N^p)}(t_1,\dots, t_p)\\
\textit{(b)}\quad\mathbb{E}^{(L)}_N[\prod_{j=1}^p\det(t_j-M)] &= \left(\prod_{j=N}^{p+N-1}(-1)^jj!\right)\mathcal{L}^{(\gamma)}_{(N^p)}(t_1,\dots,t_p)\\
\textit{(c)}\quad\mathbb{E}^{(J)}_N[\prod_{j=1}^p\det(t_j-M)] &= \left(\prod_{j=N}^{p+N-1}(-1)^jj!\frac{\Gamma(j+\gamma_1+\gamma_2+1)}{\Gamma(2j+\gamma_1+\gamma_2+1)}\right)\mathcal{J}^{(\gamma_1,\gamma_2)}_{(N^p)}(t_1,\dots,t_p)
\end{split}
\eeq
\end{theorem}
Here the subscripts $(H),(L),(J)$ indicate Hermite, Laguerre and Jacobi, respectively, and $\mathcal{H}_\lambda$, $\mathcal{L}^\gamma_\lambda$, $\mathcal{J}^{(\gamma_1,\gamma_2)}_\lambda$ are multivariate polynomials orthogonal with respect to the generalised weights in \eqref{eq:mul_poly}. 

Similar to the case of the classical compact groups, correlations of traces of Hermitian ensembles can be calculated using the theory of symmetric functions. For a partition $\lambda=(\lambda_1,\lambda_2,\dots,\lambda_N)$, $\sum_j \lambda_j\leq N$, define
\beq\label{eq:C_lambda_and_G_lambda}
\begin{split}
C_\lambda(N) &= \prod_{j=1}^N\frac{(\lambda_j+N-j)!}{(N-j)!},\\
G_\lambda(N,\gamma) &= \prod_{j=1}^N\Gamma(\lambda_j+N-j+\gamma +1).
\end{split}
\eeq
The constants $C_\lambda(N)$ and $G_\lambda(N,\gamma)$ have several interesting combinatorial interpretations which are discussed in Sec.~\ref{sec:change_of_basis}.

\begin{theorem}
\label{thm:joint_mom_traces}
Let $M$ be an  $N\times N$ GUE, LUE or JUE matrix and let $\mu=(\mu_1,\dots,\mu_l)$ be a partition such that $|\mu|=\sum_{j=1}^l\mu_l\leq N$. Then
\begin{enumerate}[label=(\alph*)]
\item 
\beq\label{eq:mom_power_gue}
\mathbb{E}_N^{(H)}\big[\prod_{j=1}^l\Tr M^{\mu_j}\big]=
\begin{cases}
\frac{1}{2^{\frac{|\mu|}{2}}\frac{|\mu|}{2}!}\sum_{\lambda\vdash |\mu|}\chi^\lambda_{(2^{|\lambda|/2})}\chi^\lambda_\mu C_\lambda(N), \quad\text{$|\mu|$ is even},\\
0,\qquad\qquad\qquad\qquad\qquad\qquad\quad otherwise,
\end{cases}
\eeq
which is a polynomial in $N$.
\item
\beq
\mathbb{E}_N^{(L)}\big[\prod_{j=1}^l\Tr M^{\mu_j}\big]= \frac{1}{|\mu|!}\sum_{\lambda\vdash |\mu|}\frac{G_\lambda(N,\gamma)}{G_0(N,\gamma)}C_\lambda(N)\chi^\lambda_{(1^{|\lambda|})}\chi^\lambda_\mu.
\eeq
\item
\beq
\mathbb{E}_N^{(J)}\big[\prod_{j=1}^l\Tr M^{\mu_j}\big]=\sum_{\lambda\vdash |\mu|}\frac{G_\lambda(N,\gamma_1)}{G_0(N,\gamma_1)}C_\lambda(N)\chi^\lambda_\mu D_{\lambda 0}^{(J)},
\eeq
where 
\beq\label{eq:det_jac_mom_traces}
D_{\lambda 0}^{(J)} = \det\left[\mathbbm{1}_{\lambda_i-i+j\geq 0}\frac{1}{(\lambda_i-i+j)!}\frac{\Gamma(2N-2i+\gamma_1+\gamma_2+2)}{\Gamma(2N+\lambda_i-i-j+\gamma_1+\gamma_2+2)}\right]_{i,j=1,\dots, N}.
\eeq
\end{enumerate}
\end{theorem}
In the above equations $\chi^\lambda_\mu$ are the characters of the symmetric group $\mathcal{S}_m$, $m=|\lambda| = |\mu|$, associated to the $\lambda^{th}$ irreducible representation on the $\mu^{th}$ conjugacy class. 

Next, we now focus our attention on the GUE with rescaled matrices $M_R=M/\sqrt{2N}$. Define the random variables
\beq
X_k:= \Tr T_k(M_R) - \mathbb{E}^{(H)}_N[\Tr T_k(M_R)].
\eeq
Here $T_k$ is the Chebyshev polynomial of degree $k$. Johansson proved the following multi-dimensional central limit theorem for $X_k$ \cite{Johansson1998}: 
\beq\label{eq:szego_thm_corollary}
(X_1,\dots, X_{2m}) \xRightarrow[]{d} (\frac{1}{2}r_1,\dots,\frac{\sqrt{2m}}{2}r_{2m}),
\eeq 
where $r_j$ are independent standard normal random variables and $\xRightarrow[]{d}$ means convergence in distribution \cite{Johansson1998}. 

Define
\beq
\mathcal{E}_{n,k} := \mathbb{E}^{(H)}_N[X_k^n]-\left(\frac{\sqrt{k}}{2}\right)^n\mathbb{E}[r_k^n]
\eeq

The formalism that we developed to study moments of traces allow us to derive explicit estimates for the error $\mathcal{E}_{n,k}$ as a function of matrix size $N$. For rescaled Gaussian matrices, the correlators of traces are Laurant polynomials in $N$. This fact can be seen from \eqref{eq:mom_power_gue} when applied to rescaled matrices. Consequently, the moments of polynomial test functions are also Laurant polynomials in $N$. 

Let $f$ and $g$ be real or complex valued functions defined on some subset of $\mathbb{R}$.  In what follows we  will use the notations $ f(x) \lesssim g(x)$ and $f(x)= O(g(x))$ as $x \to \infty$ interchangeably.  More precisely, we write $f(x) \lesssim g(x)$ ($f(x)= O(g(x))$) if and only if there exist a $x_0$ and a positive constant $K$ such that
\beq
| f(x)| \le K |g(x)| \quad \text{for  all $x> x_0$.}
\eeq  
More in general, if $t$ belongs to the extended real line, then $f(x) \lesssim g(x)$ ($f(x) = O(g(x))$) as $x \to t$  if and only if
\beq
\limsup_{x \to t} \frac{|f(x)|}{|g(x)|} < \infty.
\eeq 

We have the following theorem for Chebyshev polynomials.


\begin{theorem}\label{thm:sub-leading}
Fix $k\in\mathbb{N}$ and  let $kn\leq N$.  With the notation introduced above the following statements hold as $N\rightarrow\infty$.
\begin{enumerate}
\item For $k$ odd and $k> 1$,
\beq\label{eq:mom_chebyshev_her_oddk}
\begin{split}
\mathcal{E}_{n,k} =
\begin{cases}
0,\hspace{9.5em}\text{if $n$ is odd},\\
 d_1(n,k)\frac{1}{N^2} + O\left(\frac{1}{N^4}\right),\quad \text{if $n$ is even},
\end{cases}
\end{split}
\eeq
where
\beq\label{eq:d1}
d_1(n,k)\lesssim \, A^{\frac{3n}{k}}\pi^{-\frac{n}{2}}2^{\frac{7nk}{8}-\frac{13n}{8}+\frac{n}{6k}}k^{\frac{3n}{8}(k+2)+\frac{n}{8}+\frac{n}{4k}}n^{\frac{3n}{8}(k+1)-\frac{k}{4}+\frac{7}{8}}e^{-\frac{n}{8}(k+1)+\frac{9n}{4}+\frac{5n}{8k}+\pi\sqrt{\frac{n}{3}(k+1)}},
\eeq
as $n \to \infty$ and $k$ fixed.

\item For $k$ even,
\beq\label{eq:mom_chebyshev_her_evenk}
\begin{split}
\mathcal{E}_{n,k} =
\begin{cases}
d_2(n,k)\frac{1}{N}+O\left(\frac{1}{N^3}\right),\hspace{1.4em} \text{if $n$ is odd},\\[5pt]
d_3(n,k)\frac{1}{N^2}+O\left(\frac{1}{N^4}\right),\quad \text{if $n$ is even},
\end{cases}
\end{split}
\eeq
where
\beq
\begin{split}
&d_2(n,k)\lesssim \, A^{\frac{3n}{k}}\pi^{-\frac{n}{2}}2^{\frac{3nk}{8}-3n+\frac{n}{6k}}k^{\frac{3nk}{8}+\frac{n}{2}+\frac{9n}{4k}}n^{\frac{3nk}{8}+\frac{2n}{k}-\frac{k}{2}-\frac{3}{8}}e^{-\frac{n}{8}(k-18)+\pi\sqrt{\frac{nk}{3}}-\frac{19n}{8k}},\\
&d_3(n,k)\lesssim \,A^{\frac{3n}{k}}\pi^{-\frac{n}{2}}2^{\frac{3nk}{8}-3n+\frac{n}{6k}}k^{\frac{3nk}{8}+\frac{n}{2}+\frac{9n}{4k}}n^{\frac{3nk}{8}+\frac{2n}{k}-\frac{k}{2}+\frac{5}{8}}e^{-\frac{n}{8}(k-18)+\pi\sqrt{\frac{nk}{3}}-\frac{19n}{8k}},
\end{split}
\eeq
as  $n \to \infty$ and $k$ fixed. Here $A=1.2824\dots$ is the Glaisher–Kinkelin constant \cite{Choi2007}.
\end{enumerate}
\end{theorem}

Along with the moments we also give an estimate for the cumulants of random variables $X_k$. Computing cumulants from \eqref{eq:mom_power_gue} is not straightforward. Instead, we employ the well established connection between correlators of traces and the enumeration of ribbon graphs to estimate the cumulants. The results are elaborated in Section \ref{sec:cumulants}.

To summarise, for a fixed $n$ and $k$, we show that the $n^{th}$ moment of $X_k$ converges to the $n^{th}$ moment of independent scaled Gaussian variable as $N^{-1}$ or $N^{-2}$ depending on the parity of $n$; and the $n^{th}$ cumulant of $X_k$ converges to 0 as $N^{n-2}$ for $n>2$.  
Theorem~\ref{thm:sub-leading} provides explicit asymptotic estimates for the rate of convergence of the moments.  


\section{Background}\label{sec:background}
Symmetric polynomials arise naturally in random matrix theory because the joint eigenvalue probability density function remains invariant under the action of the symmetric group. There has been a considerable focus on symmetric functions to study moments in various ensembles \cite{Diaconis1994,Bump2006,Mezzadri2017,Dumitriu2007}. Here we define some symmetric functions that will play a central role in our calculations and state some of their properties. 
\subsection{Review of symmetric functions}
\label{revsymfun}
A partition $\lambda$ is a sequence of non-negative integers such that $\lambda_1\geq \lambda_2\geq\dots\geq\lambda_l>0$. We call the maximum $l$ such that $\lambda_l>0$ the length of the partition $l(\lambda)$ and $|\lambda|=\sum_{i=1}^l\lambda_i$ the weight. A partition can be represented with a Young diagram which is a left adjusted table of $|\lambda|$ boxes and $l(\lambda)$ rows such that the first row contains $\lambda_1$ boxes, the second row contains $\lambda_2$ boxes, and so on. The conjugate partition $\lambda^\prime$ is defined by transposing the Young diagram of $\lambda$. 
\beq
\begin{split}
\ydiagram[]{4,2,2,1} &\quad\qquad\qquad \ydiagram[]{4,3,1,1}\\
\text{Young diagram of $\lambda$}&\qquad\text{Young diagram of $\lambda^\prime$}
\end{split}
\eeq
In the above example $\lambda=(4,2,2,1)$, $|\lambda|=9$ and $l(\lambda)=4$. We denote a sub-partition $\mu$ of $\lambda$ by $\mu\subseteq\lambda$ if the Young diagram of $\mu$ is contained in the Young diagram of $\lambda$. 

Another way to represent a partition is as follows: if $\lambda$ has $b_1$ 1's, $b_2$ 2's and so on, then $\lambda = (1^{b_1}2^{b_2}\dots k^{b_k})$. In this representation, the weight $|\lambda| = \sum_{j=1}^k jb_j$ and length $l(\lambda) = \sum_{j=1}^kb_j$. In the rest of the paper, we use both notations interchangeably and we do
not distinguish partitions that differ only by a sequence of zeros; for example, (4,2,2,1) and (4,2,2,1,0,0)
are the same partitions. We denote the empty partition by $\lambda=0$ or $\lambda=()$.

The \textit{elementary symmetric functions} $e_r(x_1,\dots, x_N)$ are defined by
\beq
e_r(x_1,\dots,x_N) = \sum_{i_1<\dots <i_r}x_{i_1}\dots x_{i_r}
\eeq
and the \textit{complete symmetric functions} $h_r(x_1,\dots,x_N)$ by
\beq
h_r(x_1,\dots,x_N) = \sum_{i_1\leq\dots\leq i_r}x_{i_1}\dots x_{i_r}.
\eeq
Given a partition $\lambda$, we define
\beq
\begin{split}
e_\lambda(x_1,\dots,x_N) &= \prod_je_{\lambda_j}(x_1,\dots,x_N),\\
h_\lambda(x_1,\dots,x_N) &= \prod_jh_{\lambda_j}(x_1,\dots,x_N).
\end{split}
\eeq

The \textit{Schur polynomials} are symmetric polynomials indexed by partitions. Given a partition $\lambda$ such that $l(\lambda)\leq N$, we write
\beq\label{eq:schur_poly}
\begin{split}
S_\lambda(x_1,\dots,x_N) &= \frac{\det\left[x_i^{\lambda_j+N-j}\right]_{i,j=1,\dots , N}}{\det\left[x_i^{N-j}\right]_{i,j=1,\dots, N}}\\
&= \frac{1}{\Delta(\textbf{x})}
\begin{vmatrix}
x_1^{\lambda_1+N-1} & x_2^{\lambda_1+N-1} & \dots & x_N^{\lambda_1+N-1}\\
x_1^{\lambda_2+N-2} & x_2^{\lambda_2+N-2} & \dots & x_N^{\lambda_2+N-2}\\
\vdots & \vdots &  & \vdots\\
x_1^{\lambda_N} & x_2^{\lambda_N} & \dots & x_N^{\lambda_N}\\
\end{vmatrix},
\end{split}
\eeq  
where $\Delta(\textbf{x})$ is the Vandermonde determinant:
\beq
\Delta(\textbf{x}) = \det\left[x_i^{N-j}\right]_{i,j=1,\dots ,N} = \prod_{1\leq i<j\leq N}(x_i-x_j).
\eeq
If $l(\lambda)>N$, then $S_\lambda=0$. The Jacobi-Trudi identities express Schur polynomials in terms of elementary and complete symmetric functions:
\beq\label{eq:jacobi_trudy}
S_\lambda = \det\left[h_{\lambda_i+j-i}\right]_{i,j=1,\dots ,l(\lambda)}=\det\left[e_{\lambda^\prime_i +j-i}\right]_{i,j=1,\dots ,l(\lambda^\prime)}.
\eeq

Let $\mu = 1^{b_1}2^{b_2}\dots k^{b_k}$ and write
\beq
p_j(\textbf{x}) = \sum_{i=1}^N x^j_i, \quad j\in\mathbb{N}.
\eeq
The \emph{power sum} is defined by
\beq
P_\mu = \prod_{j=1}^k p_j^{b_j}.
\eeq
Power sum  and Schur functions  are bases in the space of homogeneous symmetric polynomials and they are related by
\begin{equation}\label{eq:power_to_schur}
\begin{split}
P_\mu = \sum_{\lambda}\chi^\lambda_\mu S_\lambda, &\quad S_\lambda = \sum_{\mu}\frac{\chi^\lambda_\mu}{z_\mu}P_\mu,\\
z_\mu &= \prod_{j}j^{b_j}b_j!,
\end{split}
\end{equation}
where $\chi^\lambda_\mu$ are the characters of the symmetric group $\mathcal{S}_m$, $m=|\lambda| = |\mu|$ and $z_\mu$ is the size of the centraliser of an element of conjugacy class $\mu$.

\begin{proposition}[Cauchy Identity \cite{Macdonald1998}]
Let $t_1,t_2,\dots$ and $x_1,x_2,\dots$ be two finite or infinite sequences of independent variables. Then,
\beq
\prod_{i,j}(1-t_ix_j)^{-1} = \sum_{\lambda}S_\lambda(\textbf{t}) S_{\lambda}(\textbf{x}).
\eeq
\end{proposition}
When the sequences $t_i$ and $x_j$ are finite,
\beq\label{eq:cauchy_identity}
\prod_{i=1}^p\prod_{j=1}^q(1-t_ix_j)^{-1} = \sum_{\lambda}S_\lambda(t_1,\dots,t_p) S_{\lambda}(x_1,\dots,x_q),
\eeq
where $\lambda$ runs over all partitions of length $l(\lambda)\leq$ min($p,q$). We also have the dual Cauchy identity \cite{Macdonald1998}:
\beq\label{eq:dual_cauchy_identity}
\prod_{i=1}^p\prod_{j=1}^q(1+t_ix_j) = \sum_{\lambda}S_{\lambda}(t_1,\dots,t_p)S_{\lambda^\prime}(x_1,\dots ,x_q).
\eeq
Since $S_\lambda =0$ or $S_{\lambda^\prime} =0$ unless $l(\lambda)\leq p$ or $l(\lambda^\prime)\leq q$, $\lambda$ runs over a finite number of partitions such that the Young diagram of $\lambda$ fits inside a $p\times q$ rectangle. 

Changing $x_j\rightarrow x_j^{-1}$ in the dual Cauchy identity and simplifying the fractions gives
\beq\label{eq:dual_cauchy_identity_a}
\prod_{i=1}^p\prod_{j=1}^q(t_i+x_j) = \sum_{\lambda\subseteq (q^p)} S_\lambda(t_1,\dots,t_p)S_{\tilde{\lambda}}(x_1,\dots ,x_q),
\eeq
where $(q^p)\equiv (\underbrace{q,\dots,q}_{p})$ and $\tilde{\lambda}=(p-\lambda^\prime_q,\dots,p-\lambda^\prime_1)$. Since the Schur polynomials are homogeneous, we have
\beq
S_\mu(-x_1,\dots, -x_q) = (-1)^{|\mu|}S_\mu(x_1,\dots,x_q).
\eeq
Thus \eqref{eq:dual_cauchy_identity_a} becomes
\beq\label{eq:char_poly_identity_schur}
\prod_{i=1}^p\prod_{j=1}^q(t_i-x_j) = \sum_{\lambda\subseteq (q^p)} (-1)^{|\tilde{\lambda}|}S_\lambda(t_1,\dots,t_p)S_{\tilde{\lambda}}(x_1,\dots ,x_q).
\eeq

The Cauchy and dual Cauchy identities, combined  with  the fact that the Schur polynomials are the characters of $U(N)$, were essentials tools in the proofs of the ratios of characteristic polynomials by Bump and Gamburd~\cite{Bump2006}. In order to prove Thm.~\ref{thm:correlations_charpoly}, we use a similar approach, in which~\eqref{eq:char_poly_identity_schur} is replaced by the generalized Cauchy dual identity~\eqref{eq:new_identity}. 

\subsection{Multivariate orthogonal polynomials}
Multivariate orthogonal polynomials can be defined by the determinant formula~\cite{Sergeev2014} 
\beq\label{eq:gen_poly}
\begin{split}
\varPhi_\mu(\textbf{x}) := \frac{1}{\Delta(\textbf{x})}
\begin{vmatrix}
\varphi_{\mu_1+N-1}(x_1) & \varphi_{\mu_1+N-1}(x_2) & \dots & \varphi_{\mu_1+N-1}(x_N)\\
\varphi_{\mu_2+N-2}(x_1) & \varphi_{\mu_2+N-2}(x_2) & \dots & \varphi_{\mu_2+N-2}(x_N)\\
\vdots & \vdots & & \vdots\\
\varphi_{\mu_N}(x_1) & \varphi_{\mu_N}(x_2) & \dots & \varphi_{\mu_N}(x_N)
\end{vmatrix},
\end{split}
\eeq
where  $l(\mu)\leq N$ and $\varphi_i, i=0,1,\dotsc$, are a sequence of polynomials orthogonal with respect to the weight $w(x)$.   One  can check by straightforward substitution that, up-to a constant, the  multivariate polynomials~\eqref{eq:gen_poly} coincide with those in~\eqref{eq:mul_poly}.  When $\varphi_j$ in \eqref{eq:gen_poly} are the Hermite, Laguerre and Jacobi polynomials we have the multivariate generalizations $\mathcal{H}_\mu$, $\mathcal{L}^{(\gamma)}_\mu$ and $\mathcal{J}^{(\gamma_1,\gamma_2)}_\mu$. These polynomials can be expressed as a linear combination of  Schur polynomials, i.e.
\beq
\varPhi_\mu(\textbf{x}) = \sum_{\nu \subseteq \mu} \kappa_{\mu \nu} S_\nu(\textbf{x}).
\eeq
For the Hermite, Laguerre and Jacobi multivariate polynomials we set the leading coefficient $\kappa_{\mu\mu}$ in consistency with the definitions~\eqref{eq:classical_hlj_ortho} and \eqref{eq:gen_poly}, 
\beq\label{eq:normalisations_mops}
\begin{split}
&\kappa_{\mu\mu}^{(H)} =1, \quad \kappa_{\mu\mu}^{(L)}= \frac{(-1)^{|\lambda|+\frac{1}{2}N(N-1)}}{G_\lambda(N,0)},\\
& \kappa_{\mu\mu}^{(J)} = \frac{(-1)^{|\lambda|+\frac{1}{2}N(N-1)}}{G_\lambda(N,\gamma_1+\gamma_2)G_\lambda(N,0)}\prod_{j=1}^N\Gamma(2N+2\lambda_j-2j+\gamma_1+\gamma_2+1).
\end{split}
\eeq

The  analogy between multivariate orthogonal polynomials  and Schur  functions becomes apparent by comparing definitions~\eqref{eq:schur_poly} with~\eqref{eq:gen_poly}.  In the literature the polynomials~\eqref{eq:gen_poly} are called generalised orthogonal polynomials or multivariate orthogonal polynomials~\cite{Baker1997} as well as generalised Schur polynomials~\cite{Sergeev2014}.  The multivariate orthogonal polynomials also satisfy a generalization of the Jacobi-Trudi identities \cite{Sergeev2014} similar to \eqref{eq:jacobi_trudy}.

The classical Hermite, Laguerre and Jacobi polynomials satisfy second order Sturm Liouville problems.  Similarly, their multivariate generalizations are eigenfunctions of second-order partial differential operators, known as Calogero–Sutherland Hamiltonians,
\beq\label{eq:calegero_de}
\begin{split}
H^{(H)} &= \sum_{j=1}^N\left(\frac{\partial^2 }{\partial x^2_j} - x_j\frac{\partial}{\partial x_j}\right) +2\sum_{\substack{j,k=1\\k\neq j}}^N\frac{1}{x_j-x_k}\frac{\partial}{\partial x_j}\\
H^{(L)}&= \sum_{j=1}^N\left(x_j\frac{\partial^2}{\partial x_j^2} + (\gamma-x_j+1)\frac{\partial}{\partial x_j}\right) + 2\sum_{\substack{j,k=1\\k\neq j}}^N\frac{x_j}{x_j-x_k}\frac{\partial}{\partial x_j}\\
H^{(J)} &= \sum_{j=1}^N\left(x_j(1-x_j)\frac{\partial^2}{\partial x_j^2} + \left(\gamma_1+1-x_j(\gamma_1+\gamma_2+2)\right)\frac{\partial}{\partial x_j}\right) + 2\sum_{\substack{j,k=1\\k\neq j}}^N\frac{x_j(1-x_j)}{x_j-x_k}\frac{\partial}{\partial x_j}
\end{split}
\eeq
These generalised orthogonal polynomials obey similar properties to their univariate counterparts~\cite{Baker1997}. The differential equations in \eqref{eq:calegero_de} are also related to the Dyson Brownian motion.

\section{Correlation functions of characteristic polynomials}\label{sec:correlation_functions}

The main tool to compute correlations of characteristic polynomials and spectral moments is Lemma.~\ref{lemma:dualcauchy_to_mulmonicpoly} which we prove here.

\begin{proposition}[Laplace Expansion]\label{prop:laplace_expansion}
Let $\Xi_{p,q}$ consist of all permutations $\sigma\in\mathcal{S}_{p+q}$ such that
\beq
\sigma(1)< \dots< \sigma(p), \quad \sigma(p+1)<\dots<\sigma(p+q).
\eeq
Let $A=a_{ij}$ be a $(p+q)\times (p+q)$ matrix, then Laplace expansion in the first $p$ rows can be written as 
\beq
\det [a_{ij}] = \sum_{\sigma\in\Xi_{p,q}}\sgn(\sigma)
\begin{vmatrix}
a_{1,\sigma(1)} &\dots &a_{1,\sigma(p)}\\
\vdots & &\vdots\\
a_{p,\sigma(1)} &\dots &a_{p,\sigma(p)}
\end{vmatrix}
\times
\begin{vmatrix}
a_{p+1,\sigma(p+1)} &\dots &a_{p+1,\sigma(p+q)}\\
\vdots & &\vdots\\
a_{p+q,\sigma(p+1)} &\dots &a_{p+q,\sigma(p+q)}
\end{vmatrix}.
\eeq
\end{proposition}
\begin{proposition}\label{prop:combinatorial_fact}
Let $\lambda$ be a partition such that $\lambda_1\leq q$ and $\lambda_1^\prime\leq p$. Then the $p+q$ numbers 
\beq
\lambda_i+p-i\, (1\leq i\leq p) ,\quad p-1+j-\lambda_j^\prime\, (1\leq j\leq q)
\eeq
are a permutation of $\{0,\dots,p+q-1\}$.
\end{proposition}
Prop.~\ref{prop:laplace_expansion} is a fact from Linear algebra and the proof of Prop.~\ref{prop:combinatorial_fact} can be found in \cite{Macdonald1998}.

\begin{proof}[Proof of Lemma.~\ref{lemma:dualcauchy_to_mulmonicpoly}]
Let the  $\varphi_j$s be monic.  Choosing a different normalization would affect Eq.~\eqref{eq:new_identity} by an overall constant.  Using the definition of generalised polynomials, Proposition~\ref{prop:laplace_expansion} and 
Proposition~\ref{prop:combinatorial_fact}, the right-hand side of \eqref{eq:new_identity} can be written as
\beq
\frac{1}{\Delta_p(\textbf{t})}\frac{1}{\Delta_q(\textbf{x})}
\begin{vmatrix}\label{eq:dummy_1}
\varphi_{p+q-1}(t_1) &\varphi_{p+q-2}(t_1) &\dots &1\\
\vdots &\vdots & &\vdots\\
\varphi_{p+q-1}(t_p) &\varphi_{p+q-2}(t_p) &\dots &1\\
\varphi_{p+q-1}(x_1) &\varphi_{p+q-2}(x_1) &\dots &1\\
\vdots &\vdots & &\vdots\\
\varphi_{p+q-1}(x_q) &\varphi_{p+q-2}(x_q) &\dots &1
\end{vmatrix}.
\eeq
Now using column operations we  arrive at
\beq\label{eq:dummy_2}
\frac{1}{\Delta_p(\textbf{t})}\frac{1}{\Delta_q(\textbf{x})}
\begin{vmatrix}
t^{p+q-1}_1 &t^{p+q-2}_1 &\dots &1\\
\vdots &\vdots & &\vdots\\
t^{p+q-1}_p &t^{p+q-2}_p &\dots &1\\
x^{p+q-1}_1 &x^{p+q-2}_1 &\dots &1\\
\vdots &\vdots & &\vdots\\
x^{p+q-1}_q &x^{p+q-2}_q &\dots &1
\end{vmatrix}.
\eeq
The determinant in \eqref{eq:dummy_2} can be evaluated using the formula for the Vandermonde determinant.  We have 
\beq\label{eq:dummy_3}
\prod_{1\leq i<j\leq p}(t_i-t_j)\prod_{1\leq i<j\leq q}(x_i-x_j)\prod_{i=1}^p\prod_{j=1}^q(t_i-x_j).
\eeq
Combining \cref{eq:dummy_1,eq:dummy_2,eq:dummy_3} proves the lemma.
\end{proof}
If $\varphi_j(-x) = (-1)^j\varphi_j(x)$, as for Hermite polynomials,  then 
\beq
\varPhi_\mu (-x_1,\dots, -x_N) = (-1)^{|\mu|}\varPhi_\mu(x_1,\dots,x_N).
\eeq
It follows that~\eqref{eq:new_identity} becomes
\beq
\prod_{i=1}^p\prod_{j=1}^q(t_i+x_j) = \sum_{\lambda\subseteq (q^p)} \varPhi_\lambda(t_1,\dots,t_p)\varPhi_{\tilde{\lambda}}(x_1,\dots ,x_q).
\eeq

\begin{proof}[Proof of Thm.~\ref{thm:correlations_charpoly}]
Unlike Hermite polynomials, the univariate Laguerre and Jacobi polynomials that obey \eqref{eq:classical_hlj_ortho} are not monic.  This fact is reflected in the normalisation in \eqref{eq:normalisations_mops} and also in the following formulae,
\beq
\begin{split}
\prod_{i=1}^p\prod_{j=1}^N(t_i-x_j) &= \sum_{\lambda\subseteq (N^p)} (-1)^{|\tilde{\lambda}|}\mathcal{H}_\lambda(t_1,\dots,t_p)\mathcal{H}_{\tilde{\lambda}}(x_1,\dots ,x_N)\\
\prod_{i=1}^p\prod_{j=1}^N(t_i-x_j) &= \left(\prod_{j=0}^{p+N-1}(-1)^jj!\right)\sum_{\lambda\subseteq (N^p)} (-1)^{|\tilde{\lambda}|}\mathcal{L}^{(\gamma)}_\lambda(t_1,\dots,t_p)\mathcal{L}^{(\gamma)}_{\tilde{\lambda}}(x_1,\dots ,x_N)\\
\prod_{i=1}^p\prod_{j=1}^N(t_i-x_j) &= \left(\prod_{j=0}^{p+N-1}(-1)^jj!\frac{\Gamma(j+\gamma_1+\gamma_2+1)}{\Gamma(2j+\gamma_1+\gamma_2+1)}\right)\\
&\qquad\times\sum_{\lambda\subseteq (N^p)} (-1)^{|\tilde{\lambda}|}\mathcal{J}^{(\gamma_1,\gamma_2)}_\lambda(t_1,\dots,t_p)\mathcal{J}^{(\gamma_1,\gamma_2)}_{\tilde{\lambda}}(x_1,\dots ,x_N)
\end{split}
\eeq
After taking the expectation value, the non-zero contribution comes from $\tilde{\lambda}=0$ because of \eqref{eq:mul_poly}. Thus $\lambda^\prime=(p^N)$ which implies $\lambda=(N^p)$.
Now using
\beq
\mathcal{H}_0=1,\quad\mathcal{L}^{(\gamma)}_0=\prod_{j=0}^{N-1}\frac{(-1)^j}{j!}, \quad \mathcal{J}^{(\gamma_1,\gamma_2)}_0=\prod_{j=0}^{N-1}\frac{(-1)^j}{j!}\frac{\Gamma(2j+\gamma_1+\gamma_2+1)}{\Gamma(j+\gamma_1+\gamma_2+1)},
\eeq
proves the result.
\end{proof}

\section{Correlations}\label{sec:moments}
In this section we calculate moments of traces and characteristic polynomials of $N\times N$ GUE, LUE and JUE matrices. We can write
\beq
\label{joint_traces}
P_\mu(\textbf{x}) = \prod_{j}(\Tr(M^j))^{b_j},
\eeq
where $\textbf{x}=(x_1,\dotsc,x_N)$ are the eigenvalues of $M$.  The power symmetric functions form a basis in the space of symmetric polynomials of degree $|\mu|$.   The main idea is to express them in the basis of the multivariate orthogonal polynomials.



\subsection{Change of basis between symmetric functions}
\label{sec:change_of_basis}
In this section we give expressions for change of basis between multivariate orthogonal polynomials and other symmetric functions. We mainly focus on the GUE but the same approach can be used for the LUE and the JUE.

\noindent\textit{Gaussian Ensemble}. Let $M$ be an $N\times N$ GUE matrix. The \textit{j.p.d.f.} of the eigenvalues is
\beq
\begin{split}
\rho^{(H)}(x_1,\dots,x_N) &= \frac{1}{Z_N^{(H)}}\Delta^2(\textbf{x})\prod_{i=1}^Ne^{-\frac{x^2_i}{2}},\\
Z_N^{(H)} &= (2\pi)^{\frac{N}{2}}\prod_{j=1}^Nj!.
\end{split}
\eeq
Denote by $H_n(x)$ the Hermite polynomials normalised according to \eqref{eq:classical_hlj_ortho1}.
Given a partition $\lambda$ with $l(\lambda)\le N$, the multivariate Hermite polynomials are given by
\beq\label{eq:mulher_def}
\begin{split}
\mathcal{H}_\lambda(\textbf{x}) = \frac{1}{\Delta(\textbf{x})}
\begin{vmatrix}
H_{\lambda_1+N-1}(x_1) & H_{\lambda_1+N-1}(x_2) & \dots & H_{\lambda_1+N-1}(x_N)\\
H_{\lambda_2+N-2}(x_1) & H_{\lambda_2+N-2}(x_2) & \dots & H_{\lambda_2+N-2}(x_N)\\
\vdots & \vdots & & \vdots\\
H_{\lambda_N}(x_1) & H_{\lambda_N}(x_2) & \dots & H_{\lambda_N}(x_N)
\end{vmatrix}
\end{split}
\eeq
and satisfy the orthogonality relation
\beq\label{eq:mulher_norm}
\begin{split}
\left\langle \mathcal{H}_\lambda,\mathcal{H}_\mu\right\rangle:=&\frac{1}{Z^{(H)}_N}\int_{\mathbb{R}^N}\mathcal{H}_\lambda(\textbf{x}) \mathcal{H}_\mu(\textbf{x})\Delta^2(\textbf{x})\prod_{i}e^{-\frac{x_i^2}{2}}\,dx_i = C_\lambda(N) \delta_{\mu\lambda},\\
&\qquad\qquad\quad C_\lambda(N) = \prod_{i=1}^N\frac{(\lambda_i+N-i)!}{(N-i)!}.
\end{split}
\eeq
Since $\lambda_i\geq 0$, the constant $C_\lambda(N)$ is a polynomial in $N$ of degree $|\lambda|$. It turns out that it has a nice interpretation in terms of characters of the symmetric group. Let $(i,j)\in\lambda$, $1 \le  j\le \lambda_i$ denote a node in the Young diagram of $\lambda$. The roots of $C_\lambda(N)$ are $i-j$,  where $i$ runs across the rows from top to bottom and $j$ across the columns from left to right of the Young diagram. For example, if $\lambda = (4,3,3,1)$, the roots of $C_\lambda(N)$ are 
\beq
\ytableausetup{centertableaux}
\begin{ytableau}
0 & -1 & -2 & -3\\
1 & 0 & -1\\
2 & 1 & 0\\
3
\end{ytableau}
\eeq
It is shown in \cite{Keating2010} that
\beq\label{eq:C_lambda}
\begin{split}
C_\lambda(N) &=  \prod_{j=1}^{l(\lambda)}\frac{(\lambda_j+N-j)!}{(N-j)!} = \prod_{(i,j)\in\lambda}(N-i+j)\\
&=\frac{|\lambda|!}{\dim V_\lambda}\sum_{\mu\vdash|\lambda|}\frac{\chi^\lambda_\mu}{z_\mu}N^{l(\mu)} = |\lambda|!\frac{S_\lambda(1^N)}{\dim\text{$V_\lambda$}}.
\end{split}
\eeq
The constant  $z_\lambda$ is defined in \eqref{eq:power_to_schur} and  $\dim V_\lambda$ is the dimension  of the irreducible representation labelled  by $\lambda$ of the symmetric group   $\mathcal{S}_{|\lambda|}$.
\beq\label{eq:dimV_lambda}
\text{dim $V_\lambda$} = |\lambda|!\frac{\prod_{1\leq j<k\leq l(\lambda)}(\lambda_j-\lambda_k-j+k)}{\prod_{j=1}^{l(\lambda)}(\lambda_j+l(\lambda)-j)!},
\eeq
and
\beq\label{eq:schur_at_ones}
S_\lambda(1^N) = \prod_{1\leq j<k \leq N}\frac{\lambda_j-\lambda_k-j+k}{k-j}.
\eeq

Schur polynomials can be expressed in terms of multivariate Hermite polynomials,
\beq\label{eq:schur_to_mulher}
S_\lambda  =\sum_{\nu\subseteq\lambda}\psi^{(H)}_{\lambda\nu}\mathcal{H}_\nu =\sum_{j=0}^{\lfloor\frac{|\lambda|}{2}\rfloor}\sum_{\nu\vdash g(j)}\psi^{(H)}_{\lambda\nu} \mathcal{H}_\nu, \quad g(j) = \begin{cases}
2j, \qquad\,\,\,\, |\lambda|\,\,\text{is even},\\
2j+1, \quad |\lambda|\,\,\text{is odd}.
\end{cases}
\eeq
The function $g(j)$ takes care of the fact that polynomials of odd and even degree do not mix similar to the one variable case. The first summation in \eqref{eq:schur_to_mulher} running over all lower order partitions takes care of the fact that $\mathcal{H}_\lambda$ are, unlike $S_\lambda$, not homogeneous polynomials. For example, when $|\lambda|$ is even,  the only partitions that appear in \eqref{eq:schur_to_mulher} are those with weight $|\nu| = |\lambda| - 2k$, $k=0, \dots, \frac{|\lambda|}{2}$, and $\nu\subseteq\lambda$.
The following proposition gives an explicit expression for the coefficients $\psi^{(H)}_{\lambda\nu}$.

\begin{proposition}\label{prop:schur_to mulher_coeff}
If $\lambda$ is a partition of length $L$ and $\nu$ is a sub-partition of $\lambda$ such that $|\lambda|-|\nu|=\text{0 mod 2}$ and $N\geq L$, then $\psi^{(H)}_{\lambda\nu}$ is the following  polynomial in $N$:
\beq\label{eq:schur_to_mulher_coef}
\begin{split}
\psi^{(H)}_{\lambda\nu} &= \frac{1}{2^{\frac{|\lambda|-|\nu|}{2}}}D_{\lambda\nu}^{(H)}\prod_{j=1}^L\frac{(\lambda_j+N-j)!}{(\nu_j+N-j)!},
\end{split}
\eeq
where
\beq
D_{\lambda\nu}^{(H)} = \det\left[\mathbbm{1}_{\lambda_j-\nu_k-j+k= \text{0 mod 2}}\left(\left(\frac{\lambda_j-\nu_k-j+k}{2}\right)!\right)^{-1}\right]_{j,k=1,\dots ,L}.
\eeq
\end{proposition}
\begin{proof}
Let $\lambda = (\lambda_1,\dots,\lambda_L,0,\dots,0)$ and $\nu=(\nu_1,\dots,\nu_l,0,\dots,0)$. Here $l$ is the length of $\nu$ and $N-l$ is the length of the sequence of zeros added to $\nu$. From \eqref{eq:mulher_norm} and the fact that $\nu\subseteq\lambda$, $l\leq L$, it follows that 
\beq
\begin{split}
\psi^{(H)}_{\lambda\nu} = & \frac{\left\langle S_\lambda,\mathcal{H}_\nu\right\rangle}{\left\langle\mathcal{H}_\nu,\mathcal{H_\nu} \right\rangle} = \frac{1}{Z^{(H)}_N\left\langle \mathcal{H}_\nu,\mathcal{H}_\nu\right\rangle}\int_{\mathbb{R}^N}S_\lambda(\textbf{x})\mathcal{H}_\nu(\textbf{x})\Delta_N^2(\textbf{x})\prod_{i=1}^Ne^{-\frac{x_i^2}{2}}\,dx_i,\\
=&\frac{1}{Z^{(H)}_N\left\langle \mathcal{H}_\nu,\mathcal{H}_\nu\right\rangle}\int_{\mathbb{R}^N}\prod_{i=1}^Ne^{-\frac{x_i^2}{2}}\,dx_i\\
&\quad\times
\begin{vmatrix}
x_1^{\lambda_1+N-1} &\dots &x_1^{\lambda_L+N-L} &H_{N-L-1}(x_1) &\dots &1\\
x_2^{\lambda_1+N-1} &\dots &x_2^{\lambda_L+N-L} &H_{N-L-1}(x_2) &\dots &1\\
\vdots & &\vdots &\vdots & &\vdots\\
x_N^{\lambda_1+N-1} &\dots &x_N^{\lambda_L+N-L} &H_{N-L-1}(x_N) &\dots &1
\end{vmatrix}\\
&\quad\times
\begin{vmatrix}
H_{\nu_1+N-1}(x_1) &\dots &H_{\nu_l+N-l}(x_1) &H_{N-l-1}(x_1) &\dots &1\\
H_{\nu_1+N-1}(x_2) &\dots &H_{\nu_l+N-l}(x_2) &H_{N-l-1}(x_2)&\dots &1\\
\vdots & &\vdots &\vdots & &\vdots\\
H_{\nu_1+N-1}(x_N) &\dots &H_{\nu_l+N-l}(x_N) &H_{N-l-1}(x_N)&\dots &1
\end{vmatrix}.
\end{split}
\eeq
The last $N-L$ and $N-l$ columns in $S_\lambda$ and in $\mathcal{H}_\nu$, respectively, are written in terms of the Hermite polynomials using column operations. In addition,  $\psi^{(H)}_{\lambda\nu}$ can be expanded as a sum over the permutations of $N$:
\beq
\begin{split}
&\psi^{(H)}_{\lambda\nu}\\
&=\frac{1}{Z^{(H)}_N\left\langle \mathcal{H}_\nu,\mathcal{H}_\nu\right\rangle}\sum_{\sigma\in\mathcal{S}_N}\text{sgn}(\sigma)\int_{\mathbb{R}^N}\prod_{i=1}^Ne^{-\frac{x_i^2}{2}}\,dx_i\,\left(x^{\lambda_1+N-1}_{\sigma(1)}\dots x^{\lambda_L+N-L}_{\sigma(L)}H_{N-L-1}(x_{\sigma(N-L-1)})\dots H_0(x_{\sigma(0)})\right)\\
&\qquad\qquad\qquad\qquad\qquad\qquad\qquad\times
\begin{vmatrix}
H_{\nu_1+N-1}(x_1) &\dots &H_{\nu_l+N-l}(x_1) &H_{N-l-1}(x_1) &\dots &1\\
H_{\nu_1+N-1}(x_2) &\dots &H_{\nu_l+N-l}(x_2) &H_{N-l-1}(x_2)&\dots &1\\
\vdots & &\vdots &\vdots & &\vdots\\
H_{\nu_1+N-1}(x_N) &\dots &H_{\nu_l+N-l}(x_N) &H_{N-l-1}(x_N)&\dots &1
\end{vmatrix}.
\end{split}
\eeq
Since the integrand is symmetric in $x_i$, every term in the above sum gives the same contribution and it is sufficient to consider only the identity permutation. All the factors can be absorbed into the determinant by multiplying the $j^{th}$ row with $x_j^{\lambda_j+N-j}$ if $j\leq L$ and with $H_{N-j}(x_{N-j})$ if $N\geq j>L$.
Then using orthogonality of Hermite polynomials \eqref{eq:classical_hlj_ortho1} for the last $N-L$ rows gives
\beq
\begin{split}
\psi^{(H)}_{\lambda\nu} &= \frac{N!}{Z^{(H)}_N\left\langle \mathcal{H}_\nu,\mathcal{H}_\nu\right\rangle}(2\pi)^{\frac{N-L}{2}}\prod_{i=L+1}^N(N-i)!\,\text{det}\left[\int_{\mathbb{R}}x_j^{\lambda_j+N-j}H_{\nu_k+N-k}(x_j)e^{-\frac{x^2_j}{2}\,dx_j}\right]_{j,k=1,\dots , L}.
\end{split}
\eeq
Expanding monomials in terms of Hermite polynomials with the formula
\beq\label{eq:mono_to_her}
x^n = n!\sum_{m=0}^{\lfloor\frac{n}{2}\rfloor}\frac{1}{2^mm!(n-2m)!}H_{n-2m}(x)
\eeq
and using orthogonality leads to \eqref{eq:schur_to_mulher_coef}. The determinant $D_{\lambda\nu}^{(H)}$ is independent of $N$ and $\psi^{(H)}_{\lambda\nu}$ is a polynomial in $N$, since $\nu\subseteq\lambda$.
\end{proof}

\begin{corollary}\label{corollary:roots_of_coeff_schur_to_mulher}
The roots of coefficients $\psi^{(H)}_{\lambda\nu}$ are integers given by the content of the skew diagram $\lambda/\nu$.
\end{corollary}
\begin{proof}
The skew diagram $\lambda/\nu$ is the set-theoretic difference of the Young diagrams of $\lambda$ and $\nu$: the set of squares that belong to the diagram of $\lambda$ but not to that of $\nu$. Using \eqref{eq:C_lambda},
\beq
\psi^{(H)}_{\lambda\nu} = \frac{1}{2^{\frac{|\lambda|-|\nu|}{2}}}\frac{C_\lambda(N)}{C_\nu(N)}D_{\lambda\nu}^{(H)}.
\eeq
Since $\nu\subseteq\lambda$, the roots of $\psi^{(H)}_{\lambda\nu}$ are integers and can be read from the skew diagram $\lambda/\nu$ whenever $D_{\lambda\nu}^{(H)}\neq 0$. For example, if $\lambda=(4,1,1)$ and $\nu = (2)$, then the roots of $\psi^{(H)}_{\lambda\nu}$ are $\{-3,-2,1,2\}$:
\beq
\ytableausetup{nosmalltableaux}
\begin{ytableau}
*(gray) 0 & *(gray) -1 & *(white) -2 &*(white) -3 \\
*(white) 1 \\
*(white) 2
\end{ytableau}
\eeq
\end{proof}

\begin{corollary}
The coefficient $\psi^{(H)}_{\lambda\lambda}=1$.
\end{corollary}
\begin{proof}
If $\nu=\lambda$,
	\beq
	\begin{split}
	\psi^{(H)}_{\lambda\lambda} &= \frac{N!}{Z^{(H)}_N\left\langle \mathcal{H}_\lambda,\mathcal{H}_\lambda\right\rangle}\det\left[\int_{\substack{\mathbb{R}}}x_j^{\lambda_j+N-j}H_{\lambda_k+N-k}(x_j)e^{-\frac{x_j^2}{2}}\,dx_j\right]_{j=1,\dots, N}.
	\end{split}
	\eeq
	By expanding monomials in terms of Hermite polynomials only the diagonal terms survive.
\end{proof}

\begin{proposition}\label{prop:psi_lambda0}
The coefficient 
\beq\label{eq:psi_lambda0}
\psi^{(H)}_{\lambda 0} = 
\begin{cases}
\frac{C_\lambda(N)}{2^{\frac{|\lambda|}{2}}\frac{|\lambda|}{2}!}\chi^\lambda_{(2^{|\lambda|/2})}, \quad \text{$|\lambda|$ is even,}\\
0, \qquad\qquad\qquad \text{$|\lambda|$ is odd,}
\end{cases}
\eeq
where $\chi^\lambda_{(2^{|\lambda|/2})}$ is the character of $\lambda$th irreducible representation evaluated on the elements of cycle-type $(2^{|\lambda|/2})$.
\end{proposition}
\begin{proof}
Since Hermite polynomials of odd and even degree do not mix, $\psi^{(H)}_{\lambda 0}=0$ when $|\lambda|$ is odd. When $|\lambda|$ is even,
\beq
D_{\lambda 0}^{(H)} =\det\left[\mathbbm{1}_{\lambda_j-j+k =\text{0 mod 2}}\,\frac{1}{\left(\frac{\lambda_j-j+k}{2}\right)!}\right].
\eeq
Let $n=|\lambda|/2$ and $L=l(\lambda)$. Let $g(x_1,\dots ,x_L)$ be a formal power series in variables $x_i$, and  $(k_1,\dots,k_L)$ be a partition constructed from $\lambda$ such that $k_j=\lambda_j+L-j$, $j=1,\dots,L$. Let
\beq
[g(x_1,\dots ,x_L)]_{(k_1,\dots,k_L)} = \text{coefficient of $x_1^{k_1}\dots x_L^{k_L}$}.
\eeq
Using Frobenius formula for characters of the symmetric group
\beq
\begin{split}
\chi^\lambda_{(2^n)} &= \left[\Delta(x_1,\dots,x_L)(x_1^2+\dots +x_L^2)^n\right]_{(k_1,\dots,k_L)}\\
&=\sum_{n_1+\dots +n_L=n}\frac{n!}{n_1!\dots n_L!}\left[\det\left[x_i^{L-j}\right]x_1^{2n_1}x_2^{2n_2}\dots x_L^{2n_L}\right]_{(\lambda_1+L-1,\lambda_2+L-2,\dots,\lambda_L)}.
\end{split}
\eeq
After absorbing $x_i^{2n_i}$ into the $i^{th}$ row of the determinant, for each $n_i$ at most one term in the $i^{th}$ row has the exponent $\lambda_i+L-i$, say the $(i,j)^{th}$ element $x_i^{2n_i+L-j}$, which implies $2n_i = \lambda_i-i+j$. For $L$-tuples $\{n_1,\dots,n_L\}$ such that there is exactly one term in each row that has the required exponent, the non-zero summands are given by $n!\,\text{sgn}(\sigma)\prod_i((\lambda_i-i+\sigma(i)/2)!)^{-1}$ where $\sigma\in\mathcal{S}_L$. Considering all such $L$-tuples and using Laplace expansion for the determinant
 proves the proposition. 
\end{proof}
Therefore the expansion of Schur polynomials in terms of multivariate Hermite polynomials can be written as 
\beq
\begin{split}
S_\lambda(x_1,\dots,x_N) &=C_\lambda(N)\sum_{\nu\subseteq\lambda}\frac{1}{2^{\frac{|\lambda|-|\nu|}{2}}}\frac{1}{C_\nu(\lambda)}D_{\lambda\nu}^{(H)}\mathcal{H}_\nu(x_1,\dots,x_N).
\end{split}
\eeq
In a similar way, by expanding Hermite polynomials in terms of monomials in the definition of $\mathcal{H}_\lambda$, multivariate Hermite polynomials can be written in the Schur basis as follows:
\beq
\mathcal{H}_\lambda = \sum_{\nu\subseteq\lambda}\kappa^{(H)}_{\lambda\nu}S_\nu = \sum_{j=0}^{\lfloor\frac{|\lambda|}{2}\rfloor}\sum_{\nu\vdash g(j)}\kappa_{\lambda\nu}^{(H)} S_\nu, \quad g(j) = \begin{cases}
2j, \qquad\,\,\,\, |\lambda|\,\,\text{is even},\\
2j+1, \quad |\lambda|\,\,\text{is odd},
\end{cases}
\eeq
where 
\beq\label{eq:mulher_to_schur_coeff}
\kappa_{\lambda\nu}^{(H)}=\left(\frac{-1}{2}\right)^{\frac{|\lambda|-|\nu|}{2}}D_{\lambda\nu}^{(H)}\prod_{j=1}^L\frac{(\lambda_j+N-j)!}{(\nu_j+N-j)!}.
\eeq
Alternatively,
\beq\label{eq:mulher_to_schur}
\begin{split}
\mathcal{H}_\lambda(x_1,\dots,x_N) &= C_\lambda(N)\sum_{\nu\subseteq\lambda} \left(\frac{-1}{2}\right)^{\frac{|\lambda|-|\nu|}{2}}\frac{1}{C_\nu(N)}D_{\lambda\nu}^{(H)}S_\nu(x_1,\dots,x_N),
\end{split}
\eeq
where $|\lambda|-|\nu|= \text{0 mod 2}$.

\noindent\textit{Laguerre ensemble.} Let $M$ be an $N\times N$ LUE matrix with eigenvalues $x_1,\dots,x_N$. For $\gamma>-1$, the \textit{j.p.d.f.} of eigenvalues is 
\beq
\begin{split}
\rho^{(L)}(x_1,\dots,x_N) &= \frac{1}{Z^{(L)}_N}\Delta^2(\textbf{x})\prod_{i=1}^Nx_i^\gamma e^{-x_i},\\
Z^{(L)}_N &= N!G_0(N,\gamma)G_0(N,0),
\end{split}
\eeq
where $G_\lambda(N,\gamma)$ is given in \eqref{eq:C_lambda_and_G_lambda}. 

The multivariate Laguerre polynomials defined by
\beq\label{eq:mullag_def}
\begin{split}
\mathcal{L}^{(\gamma)}_\lambda(\textbf{x}) = \frac{1}{\Delta_N}
\begin{vmatrix}
L^{(\gamma)}_{\lambda_1+N-1}(x_1) & L^{(\gamma)}_{\lambda_1+N-1}(x_2) & \dots & L^{(\gamma)}_{\lambda_1+N-1}(x_N)\\
L^{(\gamma)}_{\lambda_2+N-2}(x_1) & L^{(\gamma)}_{\lambda_2+N-2}(x_2) & \dots & L^{(\gamma)}_{\lambda_2+N-2}(x_N)\\
\vdots & \vdots & & \vdots\\
L^{(\gamma)}_{\lambda_N}(x_1) & L^{(\gamma)}_{\lambda_N}(x_2) & \dots & L^{(\gamma)}_{\lambda_N}(x_N)
\end{vmatrix},
\end{split}
\eeq
$l(\lambda)\leq N$, satisfy the orthogonality relation
\beq\label{eq:mullag_norm}
\begin{split}
\left\langle \mathcal{L}^{(\gamma)}_\lambda,\mathcal{L}^{(\gamma)}_\mu\right\rangle:=&\frac{1}{Z^{(L)}_N}\int_{\mathbb{R}_+^N}\mathcal{L}^{(\gamma)}_\lambda(\textbf{x}) \mathcal{L}^{(\gamma)}_\mu(\textbf{x})\Delta^2(\textbf{x})\prod_{i=1}^Nx_i^\gamma e^{-x}\,dx_i \\
=& \frac{G_\lambda(N,\gamma)}{G_0(N,\gamma)}\frac{1}{G_\lambda(N,0)}\frac{1}{G_0(N,0)}\delta_{\lambda\mu}.
\end{split}
\eeq
The polynomials in the determinant~\eqref{eq:mullag_def} are normalized according to~\eqref{eq:classical_hlj_ortho2}. The  Schur polynomials can be expanded in terms of multivariate Laguerre polynomials as
\beq\label{eq:schur_to_mullag}
S_\lambda = \sum_{\nu\subseteq\lambda}\psi^{(L)}_{\lambda\nu}\mathcal{L}^{(\gamma)}_\nu,
\eeq
where 
\beq\label{eq:schur_to_mullag_coef}
\begin{split}
\psi_{\lambda\nu}^{(L)} &= (-1)^{|\nu|+\frac{1}{2}N(N-1)}\frac{G_\lambda(N,\gamma)}{G_\nu(N,\gamma)}G_\lambda(N,0)D^{(L)}_{\lambda\nu}\\
D^{(L)}_{\lambda\nu}&= \det\left[\mathbbm{1}_{\lambda_i-\nu_j-i+j\geq 0\frac{1}{(\lambda_i-\nu_j-i+j)!}}\right]_{i,j=1,\dots ,l(\lambda)}.
\end{split}
\eeq
The coefficients $\psi^{(L)}_{\lambda\nu}$ in \eqref{eq:schur_to_mullag_coef} can be computed in a similar way as in Prop.~\ref{prop:schur_to mulher_coeff}. It is interesting to note that the quantity $|\lambda/\nu|!D^{(L)}_{\lambda\nu}$ gives the number of standard Young tableaux (SYT) of shape $\lambda/\nu$ \cite[p.344]{Stanley1999}. 

Multivariate Laguerre polynomials can also be expanded in the Schur basis:
\beq\label{eq:mullag_to_schur}
\begin{split}
\mathcal{L}^{(\gamma)}_\lambda &= \sum_{\nu\subseteq\lambda}\kappa^{(L)}_{\lambda\nu}S_\nu,\\
\kappa^{(L)}_{\lambda\nu} &= (-1)^{|\nu|+\frac{1}{2}N(N-1)}\frac{G_\lambda(N,\gamma)}{G_\nu(N,\gamma)}\frac{1}{G_\nu(N,0)}D_{\lambda\nu}^{(L)}.
\end{split}
\eeq
Similar to the Hermite case, $D^{(L)}_{\lambda 0}$ turns out to be a character of the symmetric group.
\begin{proposition}
We have
\beq\label{eq:D_lambda0_lag}
D^{(L)}_{\lambda 0} = \frac{\chi^\lambda_{(1^{|\lambda|})}}{|\lambda|!}=\frac{\dim\, V_\lambda}{|\lambda|! }.
\eeq
\begin{proof}
Same as Prop.~\ref{prop:psi_lambda0}. Note that $|\lambda|! D^{(L)}_{{\lambda 0}}$ gives the number of standard Young tableaux of shape $\lambda$.
\end{proof}
\end{proposition}

\noindent\textit{Jacobi ensemble.} Let $M$ be an $N\times N$ JUE matrix with eigenvalues $x_1,\dots,x_N$. For $\gamma_1,\gamma_2 >-1$, the \textit{j.p.d.f.} of eigenvalues is 
\beq
\begin{split}
\rho^{(J)}(x_1,\dots,x_N) &= \frac{1}{Z^{(J)}_N}\Delta^2(\textbf{x})\prod_{i=1}^Nx_i^{\gamma_1}(1-x_i)^{\gamma_2},\\
Z^{(J)}_N &= N!\prod_{j=0}^{N-1}\frac{j!\,\Gamma(j+\gamma_1+1)\Gamma(j+\gamma_2+1)\Gamma(j+\gamma_1+\gamma_2+1)}{\Gamma(2j+\gamma_1+\gamma_2+2)\Gamma(2j+\gamma_1+\gamma_2+1)}.
\end{split}
\eeq
Classical Jacobi polynomials are given by
\beq
J^{(\gamma_1,\gamma_2)}_n(x) = \frac{\Gamma(n+\gamma_1+1)}{\Gamma(n+\gamma_1+\gamma_2+1)}\sum_{j=0}^n\frac{(-1)^j}{j!(n-j)!}\frac{\Gamma(n+j+\gamma_1+\gamma_2+1)}{\Gamma(j+\gamma_1+1)}x^j
\eeq
and satisfy the orthogonality relation~\eqref{eq:classical_hlj_ortho3}. 
The multivariate Jacobi polynomials are 
\beq\label{eq:muljac_def}
\begin{split}
\mathcal{J}^{(\gamma_1,\gamma_2)}_\lambda(\textbf{x}) = \frac{1}{\Delta_N}
\begin{vmatrix}
J^{(\gamma_1,\gamma_2)}_{\lambda_1+N-1}(x_1) & J^{(\gamma_1,\gamma_2)}_{\lambda_1+N-1}(x_2) & \dots & J^{(\gamma_1,\gamma_2)}_{\lambda_1+N-1}(x_N)\\
J^{(\gamma_1,\gamma_2)}_{\lambda_2+N-2}(x_1) & J^{(\gamma_1,\gamma_2)}_{\lambda_2+N-2}(x_2) & \dots & J^{(\gamma_1,\gamma_2)}_{\lambda_2+N-2}(x_N)\\
\vdots & \vdots & & \vdots\\
J^{(\gamma_1,\gamma_2)}_{\lambda_N}(x_1) & J^{(\gamma_1,\gamma_2)}_{\lambda_N}(x_2) & \dots & J^{(\gamma_1,\gamma_2)}_{\lambda_N}(x_N)
\end{vmatrix},
\end{split}
\eeq
$l(\lambda)\leq N$, and obey the orthogonality relation
\beq\label{eq:muljac_norm}
\begin{split}
\left\langle \mathcal{J}^{(\gamma_1,\gamma_2)}_\lambda,\mathcal{J}^{(\gamma_1,\gamma_2)}_\mu\right\rangle:=&\frac{1}{Z^{(J)}_N}\int_{[0,1]^N}\mathcal{J}^{(\gamma_1,\gamma_2)}_\lambda(\textbf{x}) \mathcal{J}^{(\gamma_1,\gamma_2)}_\mu(\textbf{x})\Delta^2(\textbf{x})\prod_{i=1}^Nx_i^{\gamma_1}(1-x_i)^{\gamma_2}\,dx_i \\
=&\frac{N!}{Z^{(J)}_N}\frac{G_\lambda(N,\gamma_1)G_\lambda(N,\gamma_2)}{G_\lambda(N,\gamma_1+\gamma_2)G_\lambda(N,0)}\prod_{j=1}^N(2\lambda_j +2N -2j +\gamma_1 +\gamma_2 +1)^{-1}\delta_{\lambda\mu}.
\end{split}
\eeq
The expansion of the  Schur polynomials in terms of multivariate Jacobi polynomials is
\beq\label{eq:schur_to_muljac}
S_\lambda = \sum_{\nu\subseteq\lambda}\psi^{(J)}_{\lambda\nu}\mathcal{J}^{(\gamma_1,\gamma_2)}_\nu,
\eeq
where
\beq\label{eq:schur_to_muljac_coef}
\begin{split}
\psi^{(J)}_{\lambda\nu} &= (-1)^{|\nu|+\frac{1}{2}N(N-1)}\frac{G_\lambda(N,\gamma_1)}{G_\nu(N,\gamma_1)}G_\nu(N,\gamma_1+\gamma_2)G_\lambda(N,0)\\
&\qquad\times \mathcal{D}^{(J)}_{\lambda\nu}\prod_{j=1}^N(2\nu_j+2N-2j+\gamma_1+\gamma_2+1),\\
\mathcal{D}^{(J)}_{\lambda\nu} &= \det\left[\mathbbm{1}_{\lambda_j-\nu_k-j+k\geq 0}((\lambda_j-\nu_k-j+k)!\,\Gamma(2N+\lambda_j+\nu_k-j-k+\gamma_1+\gamma_2+2))^{-1}\right]_{j,k=1}^N.
\end{split}
\eeq
When $N=1$ \eqref{eq:schur_to_muljac} coincides with the one variable analogue
\beq
x^n = n!\,\Gamma(n+\gamma_1+1)\sum_{j=0}^n\frac{(-1)^j}{(n-j)!}\frac{(2j+\gamma_1+\gamma_2+1)\Gamma(j+\gamma_1+\gamma_2+1)}{\Gamma(j+\gamma_1+1)\Gamma(n+j+\gamma_1+\gamma_2 +2)}J^{(\gamma_1,\gamma_2)}_j(x).
\eeq
Multivariate Jacobi polynomials can be expanded in Schur polynomials via
\beq
\mathcal{J}^{(\gamma_1,\gamma_2)}_\lambda = \sum_{\nu\subseteq\lambda}\kappa_{\lambda\nu}^{(J)}S_\nu,
\eeq
where 
\beq\label{eq:muljac_to_schur_coef}
\begin{split}
\kappa^{(J)}_{\lambda\nu} &= (-1)^{|\nu|+\frac{1}{2}N(N-1)}\frac{G_\lambda(N,\gamma_1)}{G_\nu(N,\gamma_1)}\frac{1}{G_\lambda(N,\gamma_1+\gamma_2)G_\nu(N,0)}\tilde{\mathcal{D}}^{(J)}_{\lambda\nu},\\
\tilde{\mathcal{D}}^{(J)}_{\lambda\nu} &= \det\left[\mathbbm{1}_{\lambda_j-\nu_k-j+k\geq 0}\frac{\Gamma(2N+\lambda_j+\nu_k-j-k+\gamma_1+\gamma_2+1)}{(\lambda_j-\nu_k-j+k)!}\right]_{j,k=1}^N.
\end{split}
\eeq

\subsection{Moments of Schur polynomials}
\textit{Gaussian case}. Similar to the moments of monomials with respect to the Gaussian weight,
\beq
\begin{split}
&\frac{1}{\sqrt{2\pi}}\int_{\mathbb{R}}x^{2n}e^{-\frac{x^2}{2}}\,dx = (-1)^nH_{2n}(0) = \frac{2n!}{2^nn!},\\
&\frac{1}{\sqrt{2\pi}}\int_{\mathbb{R}}x^{2n+1}e^{-\frac{x^2}{2}}\,dx =0,
\end{split}
\eeq
the moments of Schur polynomials associated to a partition $\lambda$ are given by
\beq\label{eq:mom_schur_gue}
\mathbb{E}_N^{(H)}[S_\lambda] =
\begin{cases}
(-1)^{\frac{|\lambda|}{2}}\mathcal{H}_\lambda(0^N), \quad \text{$|\lambda|$ is even,}\\
0, \qquad\qquad\quad\qquad \text{$|\lambda|$ is odd,}
\end{cases}
\eeq
where 
\beq
H_\lambda(0^N) = \frac{(-1)^{\frac{|\lambda|}{2}}}{2^{\frac{|\lambda|}{2}}\frac{|\lambda|}{2}!}C_\lambda(N)\chi^\lambda_{(2^{|\lambda|/2})}.
\eeq
This can be easily seen from \eqref{eq:schur_to_mulher}, \eqref{eq:psi_lambda0}, \eqref{eq:mulher_to_schur}, and the fact that $S_\lambda=1$ for $\lambda=()$ and $S_\lambda(0^N)=0$ for any non-empty partition $\lambda$. Using \eqref{eq:mulher_norm}, $\mathbb{E}^{(H)}_N[S_\lambda]$ is a polynomial in $N$ with integer roots given by the content of $\lambda$ whenever $\chi^\lambda_{2^{|\lambda|/2}}$ is non-zero.

A few examples of moments of Schur polynomials corresponding to partitions of 4 are given below. 
\beq
\begin{matrix*}[l]
\mathbb{E}^{(H)}_N\left[S_4\right] = \frac{1}{8}N(N+1)(N+2)(N+3) \quad &\mathbb{E}^{(H)}_N\left[S_{3,1}\right] = -\frac{1}{8}(N-1)N(N+1)(N+2)\\
\mathbb{E}^{(H)}_N\left[S_{2,2}\right] = \frac{1}{4}(N-1)N^2(N+1) \quad &\mathbb{E}^{(H)}_N\left[S_{2,1,1}\right] = -\frac{1}{8}(N-2)(N-1)N(N+1)\\
\mathbb{E}^{(H)}_N\left[S_{1^4}\right] = \frac{1}{8}(N-3)(N-2)(N-1)N
\end{matrix*}
\eeq
\textit{Laguerre case}. The univariate moments are 
\beq
\frac{1}{\Gamma(\gamma +1)}\int^{\infty}_0 x^{n+\gamma}e^{-x}\, dx = \frac{\Gamma(n+\gamma +1)}{\Gamma(\gamma +1)} = n!L^{(\gamma)}_n(0).
\eeq
The moments of the Schur polynomials with respect to the Laguerre weight can be computed using \eqref{eq:schur_to_mullag},
\beq
\begin{split}
\mathbb{E}^{(L)}_N[S_\lambda] &= \frac{C_\lambda(N)}{|\lambda|!}\frac{G_\lambda(N,\gamma)}{G_0(N,\gamma)}\chi^{\lambda}_{(1^{|\lambda|})}\\
 &= (-1)^{\frac{N(N-1)}{2}}G_\lambda(N,0)\mathcal{L}^{(\gamma)}_\lambda(0^N).
 \end{split}
\eeq
Like in the the Hermite case, $\mathbb{E}^{(L)}_N(S_\lambda)$ are polynomials in $N$ with roots $i-j$ and $i-j-\gamma$, where $(i,j)\in\lambda$ as discussed in
Sec.~\ref{sec:change_of_basis}.

\noindent\textit{Jacobi case.} We have
\beq
\begin{split}
\int_0^1 x^{n+\gamma_1}(1-x)^{\gamma_2}\, dx &=n!\frac{\Gamma(\gamma_1+1)\Gamma(\gamma_2+1)}{\Gamma(n+\gamma_1+\gamma_2+2)}J^{(\gamma_1,\gamma_2)}_n(0)\\
&= \frac{\Gamma(n+\gamma_1+1)\Gamma(\gamma_2+1)}{\Gamma(n+\gamma_1+\gamma_2+2)}.
\end{split}
\eeq
Similarly, 
\beq
\begin{split}
\mathbb{E}_N^{(J)}[S_\lambda] &= \frac{G_\lambda(N,\gamma_1)}{G_0(N,\gamma_1)}C_\lambda(N) D_{\lambda 0}^{(J)}\\
&= (-1)^{\frac{N(N-1)}{2}}\frac{D^{(J)}_{\lambda 0}}{\tilde{\mathcal{D}}^{(J)}_{\lambda 0}}G_\lambda(N,\gamma_1+\gamma_2)G_\lambda(N,0)\mathcal{J}^{(\gamma_1,\gamma_2)}_{\lambda}(0^N),
\end{split}
\eeq
where $D_{\lambda 0}^{(J)}$ and $\tilde{\mathcal{D}}^{(J)}_{\lambda 0}$ are given in \eqref{eq:det_jac_mom_traces} and \eqref{eq:muljac_to_schur_coef}, respectively, and 
\beq
\mathcal{J}^{(\gamma_1,\gamma_2)}_{\lambda}(0^N) = (-1)^{\frac{N(N-1)}{2}}\frac{G_\lambda(N,\gamma_1)}{G_\lambda(N,\gamma_1+\gamma_2)G_0(N,\gamma_1)G_0(N,0)}\tilde{\mathcal{D}}^{(J)}_{\lambda 0}.
\eeq

\subsection{Moments of characteristic polynomials}

The Cauchy identity can be written as 
\beq
\prod_{i=1}^q\prod_{j=1}^N\frac{1}{(T_i-x_j)} = \frac{1}{\prod_{j=1}^qT_j^N}\sum_\lambda\sum_{\mu\subseteq\lambda}\psi_{\lambda\mu}S_\lambda(T_1^{-1},\dots ,T^{-1}_q)\varPhi_\mu(x_1,\dots,x_N),
\eeq
where $\varPhi_\mu$ is one of the generalised polynomials $\mathcal{H}_\mu$, $\mathcal{L}^{(\gamma)}_\mu$ or $\mathcal{J}^{(\gamma_1,\gamma_2)}_\mu$. By using orthogonality of multivariate polynomials \eqref{eq:mulher_norm}, \eqref{eq:mullag_norm} and \eqref{eq:muljac_norm} we have the following proposition. 
\begin{proposition}
Let $t_1,\dots,t_p$ and $T_1,\dots,T_q$ be two sets of variables. Then
\beq
\begin{split}
\prod_{j=1}^p\prod_{k=1}^q\mathbb{E}^{(H)}_N\left[\frac{\det(t_j-M)}{\det(T_k-M)}\right] &= \prod_{j=1}^q\frac{1}{T_j^N}\sum_{\substack{\lambda\subseteq (N^p)\\s.t.\, \tilde{\lambda}=\nu}}\sum_{\mu}\sum_{\nu\subseteq\mu}\frac{(-1)^{|\nu|}}{2^{\frac{|\mu|-|\nu|}{2}}}C_\mu(N)D^{(H)}_{\mu\nu}\mathcal{H}_\lambda(\textbf{t})S_\mu(\textbf{T}^{-1})\\
\prod_{j=1}^p\prod_{k=1}^q\mathbb{E}^{(L)}_N\left[\frac{\det(t_j-M)}{\det(T_k-M)}\right] &=\prod_{j=N}^{p+N-1}(-1)^jj!\prod_{k=1}^q\frac{1}{T_k^N}\\
&\quad \times\sum_{\substack{\lambda\subseteq (N^p)\\s.t.\, \tilde{\lambda}=\nu}}\sum_{\mu}\sum_{\nu\subseteq\mu}\frac{G_\mu(N,\gamma)}{G_0(N,\gamma)}\frac{C_\mu(N)}{C_\nu(N)} D^{(L)}_{\mu\nu}\mathcal{L}^{(\gamma)}_\lambda(\textbf{t})S_\mu(\textbf{T}^{-1})\\
\prod_{j=1}^p\prod_{k=1}^q\mathbb{E}^{(J)}_N\left[\frac{\det(t_j-M)}{\det(T_k-M)}\right] &=\prod_{j=N}^{p+N-1}(-1)^jj!\frac{\Gamma(j+\gamma_1+\gamma_2+1)}{\Gamma(2j+\gamma_1+\gamma_2+1)}\prod_{k=0}^{N-1}\Gamma(2k+\gamma_1+\gamma_2+2)\prod_{l=1}^q\frac{1}{T_l^N}\\
&\quad\times\sum_{\substack{\lambda\subseteq (N^p)\\s.t.\, \tilde{\lambda}=\nu}} \sum_{\mu}\sum_{\nu\subseteq\mu}\frac{G_\mu(N,\gamma_1)}{G_0(N,\gamma_1)}\frac{G_\nu(N,\gamma_2)}{G_0(N,\gamma_2)}\frac{C_\mu(N)}{C_\nu(N)} \mathcal{D}^{(J)}_{\mu\nu}\mathcal{J}^{(\gamma_1,\gamma_2)}_\lambda(\textbf{t})S_\mu(\textbf{T}^{-1})
\end{split}
\eeq
\end{proposition}hhNote that the RHS is a formal power series in the variables $\mathbf{T}$. 

\begin{corollary}
Let $\lambda=(N^p)$. If $t_i=t$ in Thm.~\ref{thm:correlations_charpoly}, then
\beq
\begin{split}
\mathbb{E}^{(H)}_N\left[\left(\det(t-M)\right)^p\right] &= C_\lambda(p)\sum_{\nu\subseteq \lambda}\left(\frac{-1}{2}\right)^{\frac{|\lambda|-|\nu|}{2}}\frac{\dim V_\nu}{|\nu|!} D^{(H)}_{\lambda\nu}t^{|\nu|},\\
\mathbb{E}^{(L)}_N\left[\left(\det(t-M)\right)^p\right] &= (-1)^{p(p+N-1)}G_\lambda(p,\gamma)\frac{G_\lambda(p,0)}{G_0(p,0)}\sum_{\nu\subseteq \lambda}\frac{(-1)^{|\nu|}}{|\nu|!\,G_\nu(p,\gamma)}\dim V_\nu D^{(L)}_{\lambda\nu}t^{|\nu|},\\
\mathbb{E}^{(J)}_N\left[\left(\det(t-M)\right)^p\right] &= \left(\prod_{j=N}^{p+N-1}\frac{1}{\Gamma(2j+\gamma_1+\gamma_2+1)}\right)(-1)^{p(p+N-1)}\frac{G_\lambda(p,\gamma_1)G_\lambda(p,0)}{G_0(p,0)}\\
&\quad\times \sum_{\nu\subseteq\lambda}\frac{(-1)^{|\nu|}}{|\nu|!\,G_\nu(p,\gamma_1)}\dim V_\nu \tilde{\mathcal{D}}^{(J)}_{\lambda\nu}t^{|\nu|},
\end{split}
\eeq
where $\dim V_\nu$ is given in \eqref{eq:dimV_lambda}.
\end{corollary}
\begin{proof}
Let us consider the GUE case. We have
\beq
\mathbb{E}^{(H)}_N\left[\left(\det(t-M)\right)^p\right] = \mathcal{H}_{(N^p)}(t^p).
\eeq
Using \eqref{eq:mulher_to_schur} and calculating $C_\lambda$ in \eqref{eq:mulher_norm} for $\lambda = (N^p)$,
\beq
C_{(N^p)}(p) = \prod_{j=1}^p\frac{(N+p-j)!}{(p-j)!},
\eeq
proves the statement. Similarly, the Laguerre and Jacobi cases can be computed in an identical way.
\end{proof}

\subsection{Joint moments of traces}
Recently the study of moments and joint moments of Hermitian ensembles have attracted considerable interest~\cite{Cunden2021,Cunden2019,Dubrovin2017,Gisonni2020,Gisonni2020jacobi}. Here we give new and self contained formulae for the joint moments of unitary ensembles in terms of characters of the symmetric group. We focus on the GUE but exactly the same method applies to the LUE and JUE.

Using \eqref{eq:power_to_schur} and \eqref{eq:schur_to_mulher}, power sum symmetric polynomials can be written in terms of multivariate Hermite polynomials
\beq\label{eq:power_to_mulher}
P_\mu = \sum_{\lambda}\sum_{\nu\subseteq\lambda}\chi_\mu^\lambda\psi^{(H)}_{\lambda\nu}\mathcal{H}_\nu.
\eeq

\begin{proof}[Proof of Thm.~\ref{thm:joint_mom_traces}]
When $|\mu|$ is odd $P_\mu$ is a sum of product of monomials in $x_i$ with the degree of at least one $x_i$ being odd. Since the generalised weight $\Delta_N^2(\textbf{x})\prod_{i=1}^Ne^{-\frac{x^2_i}{2}}$ is an even function and $P_\mu(\textbf{x})$ is  odd, $\mathbb{E}_N^{(H)}\left[P_\mu\right]$ vanishes.

When $|\mu|$ is even, writing $P_\mu$ in terms of multivariate Hermite polynomials \eqref{eq:power_to_mulher} and using orthogonality of the $\mathcal{H}_\nu$ along with \eqref{eq:psi_lambda0} proves the first line of \eqref{eq:mom_power_gue}. 
\end{proof}

\begin{remark}
When $|\mu|$ is even, the orthogonality of characters indicate that $\mathbb{E}_N^{(H)}[(\Tr M^2)^{\frac{|\mu|}{2}}]$ is a polynomial in $N$ of degree $|\mu|$. The polynomial degree of all other joint moments corresponding to the partitions of $|\mu|$ is strictly less than $|\mu|$.
\end{remark}

%
\begin{corollary}\label{cor:even odd sym}
Correlators of traces in the L.H.S. of \eqref{eq:mom_power_gue} are either even or odd polynomials in $N$. More precisely, we have
\begin{center}
\begin{tabular}{c c c} 
\hline\hline
$\mathbb{E}^{(H)}_N[P_\mu]$ & $l(\mu)$ & $|\mu|/2$ \\
\hline
\multirow{2}{10em}{Even polynomial} & even & even \\ 
& odd & odd \\ 
[10pt]
\multirow{2}{10em}{Odd polynomial} & even & odd \\ 
& odd & even \\
\hline
\end{tabular}
\end{center}
\end{corollary}
\begin{proof}
Let $|\mu|$ be even. Since $\mathbb{E}^{(H)}_N[S_\mu]$ is a polynomial in $N$ of degree $|\mu|$ and the characters $\chi^\mu_\lambda$ are integers, $\mathbb{E}^{(H)}_N[P_\lambda]$ is also a  polynomial in $N$. 
Now for any partitions $\lambda$ and $\mu$,
\beq
\begin{split}
&\chi^{\lambda^\prime}_\mu = (-1)^{|\mu|-l(\mu)}\chi_\mu^\lambda,\\
&C_{\mu^\prime}(N) = C_\mu(-N).
\end{split}
\eeq
Thus,
\beq\label{eq:mom_power_sum_symrel_1}
\begin{split}
\mathbb{E}^{(H)}_N[P_\mu] &= \frac{1}{2}\sum_{\lambda}\left(\chi_\mu^\lambda\mathbb{E}^{(H)}_N[S_\lambda] + \chi^{\lambda^\prime}_\mu\mathbb{E}^{(H)}_N[S_{\lambda^\prime}]\right)\\
&= \frac{1}{2^{\frac{|\mu|+2}{2}}\frac{|\mu|}{2}!}\sum_\lambda \chi^\lambda_{(2^{|\lambda|/2})}\chi^\lambda_\mu\left(C_\lambda(N) + (-1)^{\frac{|\mu|}{2}-l(\mu)}C_\lambda(-N)\right).
\end{split}
\eeq
The corollary is proved by noticing that the symmetric and anti-symmetric combination of $C_\lambda(N)$ and $C_\lambda(-N)$ is an even and odd polynomial in $N$, respectively.
\end{proof}

Since $\mathbb{E}^{(H)}_N[P_\mu]$ are polynomials in $N$, the domain of $N$ can be analytically continued from integers to the whole complex plane. In \cite{Cunden2019}, it is shown that $\mathbb{E}^{(H)}_N[\Tr M^{2j}]$, $j\in \mathbb{N}$, are Meixner-Pollaczek polynomials which are a family of orthogonal polynomials,
\beq
\begin{split}
\mathbb{E}^{(H)}_N\left[\Tr M^{2j}\right] &= N(2j-1)!!i^{-j}\frac{1}{j+1}P^{(1)}_j\left(iN,\frac{\pi}{2}\right)
\\
& = N(2j-1)!!\pFq{2}{1}{-j,1-N}{2}{2}
\end{split}
\eeq
where $P^{(1)}_k\left(iN,\pi/2\right)$ is a Meixner-Pollaczek polynomial and ${}_{2}F_{1}(\dots)$ is a terminating hypergeometric series. These polynomials $P^{(\lambda)}_n(x,\phi)$ satisfy
\beq
\int_{-\infty}^\infty P^{(\lambda)}_m\left(x,\phi\right)P^{(\lambda)}_n\left(x,\phi\right)|\Gamma(\lambda + ix)|^2 e^{(2\phi-\pi)x}\,dx = \frac{2\pi \Gamma(n+2\lambda)}{(2\sin\phi)^{2\lambda}n!}\delta_{nm}.
\eeq
Clearly, the zeros of $\mathbb{E}^{(H)}_N\left[\Tr M^{2j}\right]$ lie on the line $\text{Re}(N)=0$
 
Correlators of traces are combinatorial objects as they are connected to enumeration of ribbon graphs \cite{Bessis1980,tHooft1974planar,tHooft1974two}. This connection is briefly discussed in Appendix~\ref{app:ribbon graphs}. By counting ribbon graphs, it can be easily shown that
\beq
\begin{split}
\mathbb{E}^{(H)}_N\left[\Tr M^{2k-1} \Tr M\right] &=(2k-1)\mathbb{E}^{(H)}_N[\Tr M^{2k-2}]\\
&= N(2k-1)!!i^{-k+1}\frac{1}{k}P_{k-1}^{(1)}\left(iN,\frac{\pi}{2}\right).
\end{split}
\eeq
Thus $\mathbb{E}^{(H)}_N[P_\mu]$, $\mu=(2k-1,1)$, is also a polynomial in $N$. A few examples of joint moments of traces corresponding to partitions of 6 are given below. Here $p_j=\Tr M^j$.
\beq\label{eq:expectedvalue_par6_gue} 
\begin{matrix*}[l]
\mathbb{E}^{(H)}_N\left[p_6\right] = 5N^2(N^2+2) &\mathbb{E}^{(H)}_N\left[p_5p_1\right] = 5N(2N^2+1)\\
\mathbb{E}^{(H)}_N\left[p_4p_2\right] = N(2N^2+1)(N^2+4) &\mathbb{E}^{(H)}_N\left[p_4p_1^2\right] = N^2(2N^2+13)\\
\mathbb{E}^{(H)}_N\left[p_3^2\right] = 3N(4N^2+1) &\mathbb{E}^{(H)}_N\left[p_3p_2p_1\right] = 3N^2(N^2+4)\\
\mathbb{E}^{(H)}_N\left[p_3p_1^3\right] = 3N(3N^2+2) &\mathbb{E}^{(H)}_N\left[p_2^3\right] = N^2(N^2+2)(N^2+4)\\
\mathbb{E}^{(H)}_N\left[p_2^2p_1^2\right] = N(N^2+2)(N^2+4) &\mathbb{E}^{(H)}_N\left[p_2p_1^4\right] = 3N^2(N^2+4)\\
\mathbb{E}^{(H)}_N\left[p_1^6\right] =15N^3
\end{matrix*}
\eeq

\section{Eigenvalue fluctuations}\label{sec:fluctuations}
\subsection{Moments}
 Here we focus on the GUE but the Laguerre and Jacobi ensembles can be studied in a similar way. Consider the rescaled GUE matrices $M_R=M/\sqrt{4N}$ of size $N$ with \textit{j.p.d.f.}
\beq\label{eq:rescaled_jpdf_her}
\frac{(4N)^{\frac{N^2}{2}}}{(2\pi)^{\frac{N}{2}}\prod_{j=1}^Nj!}\prod_{1\leq i<j\leq N}(x_i-x_j)^2\prod_{j=1}^Ne^{-2Nx_j^2}
\eeq
The limiting eigenvalue density is
\beq\label{eq:semicircle}
\rho_{sc}(x) = \frac{2}{\pi}\sqrt{1-x^2}.
\eeq

\begin{proposition}\label{prop:k=1}
 We have
\beq
\mathbb{E}^{(H)}_N\left[(\Tr M_R)^{2n}\right]=\frac{2n!}{2^{3n}n!}.
\eeq
\end{proposition}
\begin{proof}
When $\mu=(1^{2n})$ in \eqref{eq:mom_power_gue}, using \eqref{eq:C_lambda} and the fact that $\chi^\lambda_{(1^{2n})} = \dim\, V_\lambda$,
\beq
\mathbb{E}^{(H)}_N\left[(\Tr M_R)^{2n}\right] = \frac{2n!}{2^{3n}n!}\frac{1}{N^n}\sum_{\lambda\vdash 2n}\chi^{\lambda}_{(2^n)}S_\lambda(1^N).
\eeq
Using \eqref{eq:power_to_schur} and $P_\nu(1^N)=N^{l(\nu)}$, 
\beq
\begin{split}
\mathbb{E}^{(H)}_N\left[(\Tr M_R)^{2n}\right]&=
\frac{2n!}{2^{3n}n!}\frac{1}{N^n}P_{(2^n)}(1^N)=\frac{2n!}{2^{3n}n!}.
\end{split}
\eeq
\end{proof}

The R.H.S. is the $2n^{th}$ moment of $r_1/2$ where $r_{1}\sim\mathcal{N}(0,1)$. This exact equality of moments with the moments of Gaussian normals is special to $\mathbb{E}^{(H)}_N\left[(\Tr M_R)^{2n}\right]$. In general, one can consider moments of the form $\mathbb{E}_N^{(H)}[(\Tr\,g(M))^{n}]$ for a well-defined function $g$.

Johansson \cite{Johansson1998} showed that when $g$ is the Chebyshev polynomial of the first kind of degree $k$, the random variable
\beq
X_k = \Tr T_k(M_R) - \mathbb{E}^{(H)}_N[\Tr T_k(M_R)],\quad k=0,1,\dots,
\eeq
converges in distribution to the Gaussian variable $\mathcal{N}(0,k/4)$. In this section we prove Theorem \ref{thm:sub-leading}, which implies that
\beq\label{eq:moments of xk}
\mathbb{E}^{(H)}_N[X_k^n] = \left(\frac{\sqrt{k}}{2}\right)^n\frac{n!}{2^{\frac{n}{2}}\left(\frac{n}{2}\right)!}\eta_n + d(n,k)\frac{1}{N^{1+m_{k, n}}} + O(N^{-2}),
\eeq
where $\eta_n = 1$ if $n$ is even and 0 otherwise, and where $m_{k, n}$ is either 0 or 1, with asymptotic estimates for $d(n,k)$. Results for $k=1$ are already discussed in Prop.~\ref{prop:k=1}. We first consider $X_2$ and discuss results for general values of $k$ in Sec.~\ref{sec:gen degree}.

\subsubsection{Second degree}
We have that 
\beq
\mathbb{E}^{(H)}_N[(\Tr M^2_R)^n] = \frac{1}{(4N)^n}\prod_{j=0}^{n-1}(N^2+2j).
\eeq
For a fixed $n$, this can be obtained by substituting in the character values of $\mathcal{S}_{2n}$ in \eqref{eq:mom_power_gue}. Alternatively, a proof by counting topologically invariant ribbon graphs is sketched in App.~\ref{app:ribbon graphs}. Clearly,
\beq
\begin{split}
\mathbb{E}^{(H)}_N[X^n_2] &= \mathbb{E}^{(H)}_N\left[\left(2\Tr M^2_R-\frac{N}{2}\right)^n\right]\\
&=\sum_{j=0}^n\binom{n}{j}\left(-\frac{N}{2}\right)^{n-j}\mathbb{E}^{(H)}_N[(2\Tr M_R^2)^j]\\
&=\frac{N^n}{2^{n+1}}\sum_{j=0}^n(-1)^{n-j}2^jN^{2-2j}\binom{n}{j}\frac{\Gamma\left(\frac{N^2}{2}+j\right)}{\Gamma\left(\frac{N^2}{2}+1\right)}
\end{split}
\eeq
The asymptotic expansion for the ratios of Gamma functions is \cite{Frenzen1987}
\beq\label{eq:asymp ratio gamma}
\frac{\Gamma(z+a)}{\Gamma(z+b)}\sim z^{a-b}\sum_{l=0}^\infty \frac{1}{z^l}\binom{a-b}{l}B_l^{(a-b+1)}(a),\quad a,b\in\mathbb{C},\quad z\rightarrow\infty,
\eeq
where $B_j^{(l)}$ are generalised Bernoulli polynomials. Hence
\beq\label{eq:x_2}
\mathbb{E}^{(H)}_N[X_2^n] = \frac{N^n}{2^n}\sum_{j=1}^n\sum_{l=0}^{j-1}(-1)^{n-j+l}\frac{2^l}{N^{2l}}\binom{n}{j}\binom{j-1}{l}B^{(j)}_l(0).
\eeq
In arriving at \eqref{eq:x_2} we used  
\beq
B^{(j)}_l(j) = (-1)^lB^{(j)}_l(0).
\eeq
Here $B^{(j)}_l(0)$ are generalised Bernoulli numbers and the first few numbers are given below.
\beq\label{eq:gen ber num}
\begin{split}
B^{(j)}_0(0) &= 1\\
B^{(j)}_1(0) &= -\frac{j}{2}\\
B^{(j)}_2(0) &= \frac{j^2}{4}-\frac{j}{12}\\
B^{(j)}_3(0) &= -\frac{j^3}{8}+\frac{j^2}{8}.
\end{split}
\eeq
By inserting \eqref{eq:gen ber num} into \eqref{eq:x_2},
\beq
\begin{split}
\text{Coef. of $N^n$}:&\,\, \frac{1}{2^n}\sum_{j=1}^n(-1)^{n-j}\binom{n}{j}=0\\
\text{Coef. of $N^{n-2}$}:&\,\, \frac{1}{2^{n}}\sum_{j=2}^{n}(-1)^{n-j}\binom{n}{j}j(j-1)\\
&\,\, = \frac{n}{2^n}\sum_{j=2}^n(-1)^{n-j}(j-1)\binom{n-1}{j-1} = 0.
\end{split}
\eeq

Calculating the coefficient of $N^{n-2l}$ for arbitrary $n$ and $l$ is not straightforward because there are no simple expressions for generalised Bernoulli numbers. Though these numbers can be written in terms of Stirling's numbers of first kind, the coefficients can be explicitly computed only for small values of $l$. It can be shown for a given $n$ that 
\beq
\begin{split}
\text{Coef. of $N^{n-2k}$} &= 0,\quad \text{for $0\leq k< \lfloor n/2\rfloor$},\\
\text{Coef. of $N^0$} &= \frac{n!}{2^{n}\left(\frac{n}{2}\right)!}\eta_n,
\end{split}
\eeq
where $\eta_n=1$ if $n$ is even and 0 otherwise.   Our goal is not to compute these coefficients more generally, but rather to give an estimate for the sub-leading term in \eqref{eq:moments of xk}. Since the Chebyshev polynomials of even and odd degree do not mix, the moments of $X_k$ show a similar behaviour as in Corollary.~\ref{cor:even odd sym}. 
\beq
\mathbb{E}^{(H)}_N[X_2^n]=
\begin{cases}
d_2(n,2)\frac{1}{N}+O(N^{-3}),\,\qquad\quad\qquad\text{if $n$ is odd},\\
\frac{n!}{2^{n}\left(\frac{n}{2}\right)!} + d_3(n,2)\frac{1}{N^{2}} +O(N^{-4}),\quad \text{if $n$ is even}.
\end{cases}
\eeq
Coefficients $d_2(n,2)$ and $d_3(n,2)$ can be estimated using \eqref{eq:x_2}. In the next section, we give an estimate of these coefficients for arbitrary values of $n$ and $k$.

\subsubsection{General degree}\label{sec:gen degree}
The explicit expression for the joint moments of eigenvalues in Thm.~\ref{thm:joint_mom_traces} allows us to obtain Thm.~\ref{thm:sub-leading}. Consequently, $X_k$ converges to a normal random variable. For a fixed $k$ and $n$,
\beq
X_k\rightarrow \frac{\sqrt{k}}{2}\mathcal{N}(0,1)\quad \text{as}\,\, N\rightarrow\infty.
\eeq
In reality, the correct bounds in Thm.~\ref{thm:sub-leading} are much more smaller than given. This is due to sequential cancellations in the sum 
\beq
\sum_\lambda\chi^\lambda_\mu\chi^\lambda_{2^{|\mu|/2}}C_\lambda(N)
\eeq
and in the Chebyshev expansion
\beq
\Tr\, T_k(M_R)=\frac{k}{2}\sum_{j=0}^{\lfloor\frac{k}{2}\rfloor}(-1)^j\frac{(k-j-1)!}{j!(k-2j)!}2^{k-2j}M_R^{k-2j}.
\eeq
The bounds in Thm.~\ref{thm:sub-leading} are better for smaller moments.

To prove Thm.~\ref{thm:sub-leading}, we first need to estimate the coefficient of the $1/N$ term in the Laurent series of $\mathbb{E}^{(H)}_N[P_\mu]$ of rescaled matrices which leads to estimating the characters of the symmetric group.

All the characters of the symmetric group are integers and satisfy
\beq\label{eq:char trivial bound}
\frac{|\chi^\lambda_\mu|}{\chi^\lambda_{(1^{|\mu|})}}<1.
\eeq
It turns out that under suitable assumptions, the ratios $|\chi^\lambda_\mu|/\chi^\lambda_{(1^{|\mu|})}$ are very small, sometimes exponentially and super-exponentially small \cite{Fomin1997,Roichman1996}. Particularly useful bounds are 
\beq\label{eq:char bound in terms of power}
|\chi^\lambda_\mu|\leq (\chi^\lambda_{(1^{|\mu|})})^{a_\mu},
\eeq
where $a_\mu$ depends on $\mu$.

The frequency representation of a partition $\mu=(1^{b_1},2^{b_2},\dots,k^{b_k})$ also represents a permutation cycle of an element in $\mathcal{S}_{|\mu|}$. The number $b_j$ gives the number of $j-$cycles in $\mu$, $1\leq j\leq k$. For example, if $b_1=0$ then are no 1-cycles. In other words, there are no fixed points when $b_1=0$.  

The only obstruction to the small character values of $|\chi^\lambda_\mu|$ is when $\mu$ has many short cycles \cite{Larsen2008}. With this information,\begin{proposition}
\leavevmode
\begin{enumerate}[label=(\alph*)]
\item Given any $\lambda\in\text{Irr}(\mathcal{S}_n)$, let $\mu=(m^{n/m})$, then \cite{Fomin1997}
\beq\label{eq:box char}
|\chi^\lambda_\mu|\leq c\,n^{\frac{1}{2}\left(1-\frac{1}{m}\right)}\left(\chi^\lambda_{(1^{|\mu|})}\right)^{\frac{1}{m}},
\eeq
where $c$ is an absolute constant.

\item If $\mu\in\mathcal{S}_n$ is fixed-point-free, or has $n^{o(1)}$ fixed points, then for any $\lambda\in\text{Irr}(\mathcal{S}_n)$ \cite{Larsen2008},
\beq\label{eq:no fix cyc char}
|\chi^\lambda_\mu|\leq \left(\chi^\lambda_{1^{|\mu|}}\right)^{\frac{1}{2}+o(1)}.
\eeq

\item Fix $a\leq 1$ and let $\mu\in\mathcal{S}_n$ with at most $n^a$ cycles. Then for any $\lambda\in\text{Irr}(\mathcal{S}_n)$ \cite{Larsen2008},
\beq\label{eq:cycle char}
|\chi^\lambda_\mu|\leq \left(\chi^\lambda_{(1^{|\mu|})}\right)^{a+o(1)}.
\eeq
\end{enumerate}
\end{proposition}

\begin{proposition}\label{prop:p_mu coef 1/N bound}
For a given $\mu$, $\mathbb{E}^{(H)}_N[P_\mu]$ is a Laurent polynomial in $N$ with  
\beq\label{eq:p_mu coef 1/N bound}
\text{Coefficient of $1/N^q$ in $\mathbb{E}^{(H)}_N[P_\mu]$}\lesssim 2^{-\frac{|\mu|}{2}-q-\frac{3}{2}}|\mu|^{\frac{3|\mu|}{4}-\frac{11}{8}+q}e^{-\frac{|\mu|}{4}+\pi\sqrt{\frac{2}{3}|\mu|}},\quad q\in\mathbb{N},
\eeq
as $|\mu| \to \infty$.
\end{proposition}
\begin{proof}
For rescaled matrices, the expected value of $P_\mu$ is
\beq\label{eq:rescaled p_mu}
\mathbb{E}_N^{(H)}\big[\prod_{j=1}^l\Tr M_R^{\mu_j}\big]=
\begin{cases}
\frac{1}{(8N)^{\frac{|\mu|}{2}}\frac{|\mu|}{2}!}\sum_{\lambda\vdash |\mu|}\chi^\lambda_{(2^{|\lambda|/2})}\chi^\lambda_\mu C_\lambda(N), \quad\text{$|\mu|$ is even},\\
0,\hspace{14.5em} otherwise,
\end{cases}
\eeq
Using \eqref{eq:C_lambda} we obtain
\beq
\frac{\Gamma(N+1)}{\Gamma(N-|\lambda|+1)}\leq C_\lambda(N)\leq \frac{\Gamma(N+|\lambda|)}{\Gamma(N)}.
\eeq
Using the asymptotics of Gamma functions, as $N\rightarrow\infty$,
\beq
\frac{\Gamma(N+|\lambda|)}{\Gamma(N)}\sim N^{|\lambda|}\sum_{l=0}^\infty\frac{1}{N^l}\binom{|\lambda|}{l}B_l^{(|\lambda|+1)}(|\lambda|),
\eeq
where $B_l^{(j)}(x)$ are generalised Bernoulli polynomials of degree $l$. Thus, the coefficient of $1/N^q$ in \eqref{eq:mom_power_gue} is bounded by
\beq\label{eq:p_mu coef 1/N first bound}
\text{Coefficient of $1/N^q$ in $\mathbb{E}^{(H)}_N[P_\mu]$}\leq \frac{1}{8^{\frac{|\mu|}{2}}\frac{|\mu|}{2}!}\binom{|\mu|}{\frac{|\mu|}{2}+q}B^{(|\mu|+1)}_{\frac{|\mu|}{2}+q}(|\mu|)\sum_\lambda|\chi^\lambda_\mu||\chi^\lambda_{2^{|\mu|/2}}|.
\eeq
Using \eqref{eq:char bound in terms of power} and \eqref{eq:box char}, the R.H.S. of \eqref{eq:p_mu coef 1/N first bound} is bounded from above by
\beq\label{eq:p_mu coef 1/N second bound}
\frac{c}{8^{\frac{|\mu|}{2}}\frac{|\mu|}{2}!}\binom{|\mu|}{\frac{|\mu|}{2}+q}|\mu|^{\frac{1}{4}}(\chi^\lambda_{1^{|\mu|}})_{\max}^{a_\mu+\frac{1}{2}}\,\#\intpar\text{($|\mu|$)}\,B^{(|\mu|+1)}_{\frac{|\mu|}{2}+q}(|\mu|)
\eeq
The maximum of the dimension of the irreducible representation \cite{Mckay1976}
\beq\label{eq:max dim}
(\chi^\lambda_{1^{|\mu|}})_{\max}\leq (2\pi)^{\frac{1}{4}}|\mu|^{\frac{|\mu|}{2}+\frac{1}{4}}e^{-\frac{|\mu|}{2}}
\eeq
and number of partitions grow as \cite{Hardy1918,Uspensky1920}
\beq\label{eq:no of par}
\#\intpar\text{($|\mu|$)}\sim\frac{1}{4\sqrt{3}|\mu|}\exp\left(\pi\sqrt{\frac{2|\mu|}{3}}\right),\quad \text{as}\,\, |\mu|\rightarrow\infty.
\eeq
Polynomials $B^{(j)}_l(x)$ satisfy
\beq
B^{(j+1)}_l(x) = \left(1-\frac{l}{j}\right)B^{(j)}_l(x)+l\left(\frac{x}{j}-1\right)B^{(j)}_{l-1}(x).
\eeq
Hence,
\beq\label{eq:gen ber estimate}
B^{(|\mu|+1)}_{\frac{|\mu|}{2}+q}(|\mu|) = \left(\frac{1}{2}-\frac{q}{|\mu|}\right)B^{(|\mu|)}_{\frac{|\mu|}{2}+q}(|\mu|)\lesssim \frac{1}{2^{\frac{|\mu|}{2}+q+1}}|\mu|^{\frac{|\mu|}{2}+q}, \qquad \text{ as $|\mu|\to \infty$}.
\eeq
Inserting  $a_\mu=1$, \eqref{eq:max dim}, and \eqref{eq:gen ber estimate} in \eqref{eq:p_mu coef 1/N second bound} we have that
\beq
\text{coefficient of $1/N^q$ in $\mathbb{E}^{(H)}_N[P_\mu]$}\lesssim  \,\frac{1}{2^{2|\mu|+q+3}}\frac{1}{\frac{|\mu|}{2}!}\binom{|\mu|}{\frac{|\mu|}{2}+q}|\mu|^{\frac{5|\mu|}{4}-\frac{3}{8}+q}e^{-3\frac{|\mu|}{4}+\pi\sqrt{\frac{2}{3}|\mu|}},
\eeq
as $|\mu| \to \infty$. Now using Stirling's approximation proves \eqref{eq:p_mu coef 1/N bound}.
\end{proof}
We are now ready to prove Thm.~\ref{thm:sub-leading}.

\begin{proof}[Proof of Thm.~\ref{thm:sub-leading}] Using \eqref{eq:rescaled p_mu}, it can be seen that the joint moments of traces of rescaled matrices are Laurent polynomial in $N$ with rational coefficients. Thus the central moments of traces of Chebyshev polynomials are also Laurent polynomials. Since  $X_k(M_R)$ converges in distribution to a normal random variable as $N\rightarrow\infty$, $\mathbb{E}^{(H)}_N\left[X_k^n\right]$ is a polynomial in $1/N$ with constant term given in \eqref{eq:mom_chebyshev_her_oddk} and \eqref{eq:mom_chebyshev_her_evenk} depending on whether $k$ is odd and even, respectively. 

To estimate the sub-leading term in $\mathbb{E}^{(H)}_N\left[X_k^n\right]$, we consider $k$ even and odd cases separately.

\noindent(1) For $k$ odd, $\mathbb{E}^{(H)}_N\left[\Tr\,T_{k}(M_R)\right]=0$. Using the expansion of Chebyshev polynomials of the first kind,
\beq\label{eq:cheby_mu coef 1/N first bound}
\begin{split}
\mathbb{E}^{(H)}_N[X_k^n]&=\mathbb{E}^{(H)}_N\left[\left(\Tr\,T_{k}(M_R)\right)^n\right] \\
&= \mathbb{E}^{(H)}_N\left[\left(k\sum_{j=0}^{\frac{k-1}{2}}(-1)^{\frac{k-1}{2}-j}\frac{(\frac{k-1}{2}+j)!}{(\frac{k-1}{2}-j)!(2j+1)!}2^{2j}\Tr\,M^{2j+1}\right)^n\right]\\
&=k^n\sum_{n_0+\dots+n_{\frac{k-1}{2}}=n} \binom{n}{n_0,\dots,n_{\frac{k-1}{2}}}\prod_{j=0}^{\frac{k-1}{2}}(-1)^{\frac{k-1}{2}n_j-jn_j}\left(\frac{(\frac{k-1}{2}+j)!}{(\frac{k-1}{2}-j)!(2j+1)!}\right)^{n_j}2^{2jn_j}\\
&\qquad\times \mathbb{E}^{(H)}_N\left[P_\mu\right],
\end{split}
\eeq
where
\beq
\mathbb{E}^{(H)}_N[P_\mu]=\mathbb{E}^{(H)}_N\left[\prod_{l=0}^{\frac{k-1}{2}}(\Tr\,M_R^{2l+1})^{n_l}\right], \quad \mu=(1^{n_0},3^{n_1},\dots,k^{n_{\frac{k-1}{2}}}).
\eeq
The odd moments of $\Tr\,T_k(M_R)$ are identically zero because $\mathbb{E}^{(H)}_N[P_\mu]=0$ when $|\mu|$ is odd, see \eqref{eq:rescaled p_mu}. When $n$ is even, the leading term is given by the $n^{th}$ moment of $\sqrt{k}r_k/2$, $r_k\sim\mathcal{N}(0,1)$, according to Szeg\H{o}'s theorem. For $n$ even, $l(\mu)$ is always even. Hence the sub-leading term in \eqref{eq:cheby_mu coef 1/N first bound} is $O(N^{-2})$ (see Corollary.~\ref{cor:even odd sym}. Note that the matrix is now rescaled).

The maximum possible degree of $\mu$ is $|\mu|=nk$ when $n_{\frac{k-1}{2}}=n$, $n_j=0$ for $j=0,\dots,\frac{k-3}{2}$ and the minimum degree is $|\mu|=n$ when $n_0=n$, $n_j=0$ for $j=1,\dots,\frac{k-1}{2}$. The coefficient of $1/N^2$ in \eqref{eq:cheby_mu coef 1/N first bound} is estimated using \eqref{eq:p_mu coef 1/N bound} by choosing an appropriate $\mu$. Note that the multinomial coefficient is maximum when all $n_j$'s are approximately equal. In this case $\mu=(1^{\frac{2n}{k+1}},3^{\frac{2n}{k+1}},\dots,k^{\frac{2n}{k+1}})$ and $|\mu|=n(k+1)/2$. For a fixed $k$ as $n$ increases, the number of short cycles in $\mu$ increases. Using \eqref{eq:cycle char}, 
\beq
|\chi^\lambda_\mu| \leq \chi^\lambda_{1^{|\mu|}},
\eeq
which implies $a_\mu=1$ in \eqref{eq:p_mu coef 1/N bound}.

Let 
\beq
d_1(n,k) = \big[\mathbb{E}^{(H)}_N[\left(\Tr\,T_{k}(M_R)\right)^n]\big]_{1/N^2}
\eeq
denote the coefficient of $1/N^2$ in $\mathbb{E}^{(H)}_N\left[\left(\Tr\,T_{k}(M_R)\right)^n\right]$. Putting $q=2$ in \eqref{eq:p_mu coef 1/N bound},
\beq
\begin{split}
d_1(n,k)\sim &\, k^n\frac{n!}{\left(\frac{2n}{k+1}!\right)^{\frac{k+1}{2}}}\left(\prod_{j=0}^{\frac{k-1}{2}}\frac{(\frac{k-1}{2}+j)!}{(\frac{k-1}{2}-j)!(2j+1)!}\right)^{\frac{2n}{k+1}}2^{2|\mu|} \big[\mathbb{E}^{(H)}_N\left[P_\mu\right]\big]_{1/N^2}, \quad \text{as $n \to \infty$.}
\end{split}
\eeq
Now,
\beq
\begin{split}
\prod_{j=0}^{\frac{k-1}{2}}\frac{(\frac{k-1}{2}+j)!}{(\frac{k-1}{2}-j)!(2j+1)!} &= 2^{-\frac{5}{24}-\frac{1}{4}k(k+2)}e^{\frac{1}{8}}\pi^{\frac{1}{4}(k+2)}\frac{1}{A^{\frac{3}{2}}}\frac{G(k+1)}{G\left(\frac{k}{2}+2\right)G\left(\frac{k+1}{2}\right)\left(G\left(\frac{k+3}{2}\right)\right)^2}
\end{split}
\eeq
where $G(x)$ is Barnes-G function and $A=1.2824\dots$ is the Glaisher-Kinkelin constant. 

Using asymptotics of Barnes-G functions and Stirling's approximation
\beq
\begin{split}
\left(\prod_{j=0}^{\frac{k-1}{2}}\frac{(\frac{k-1}{2}+j)!}{(\frac{k-1}{2}-j)!(2j+1)!}\right)^{\frac{2n}{k+1}} &\sim\,A^{\frac{3n}{k}}\pi^{-\frac{n}{2}}2^{\frac{nk}{2}-n+\frac{n}{6k}}k^{-\frac{3n}{2}+\frac{n}{4k}}e^{\frac{9n}{4}+\frac{5n}{8k}},\\
\frac{n!}{\left(\frac{2n}{k+1}!\right)^{\frac{k+1}{2}}}&\sim \frac{1}{\pi^{\frac{k-1}{4}}}\frac{1}{2^{n+\frac{k}{2}}}n^{-\frac{k-1}{4}}(k+1)^{n+\frac{k+1}{4}}, \quad \text{as $n\to \infty$.}
\end{split}
\eeq

By combining the previous equations we arrive at
\beq\label{eq:cheby_mu coef 1/N second bound}
d_1(n,k)\lesssim A^{\frac{3n}{k}}\pi^{-\frac{1}{4}(2n+k)}2^{\frac{7nk}{8}-\frac{13n}{8}+\frac{n}{6k}-\frac{k}{2}}n^{\frac{3n}{8}(k+1)-\frac{k}{4}+\frac{7}{8}}k^{\frac{3n}{8}(k+2)+\frac{n}{8}+\frac{n}{4k}+\frac{k}{4}+\frac{7}{8}}e^{-\frac{n}{8}(k+1)+\frac{9n}{4}+\frac{5n}{8k}+\pi\sqrt{\frac{n}{3}(k+1)}},
\eeq
as  $n \to \infty$.
We are interested to find the order of the coefficient of $1/N$ as $n$ increases for a fixed $k$. To capture the right behaviour, it is sufficient to approximate \eqref{eq:cheby_mu coef 1/N second bound} to
\beq
\begin{split}
d_1(n,k)\lesssim &\, A^{\frac{3n}{k}}\pi^{-\frac{n}{2}}2^{\frac{7nk}{8}-\frac{13n}{8}+\frac{n}{6k}}k^{\frac{3n}{8}(k+2)+\frac{n}{8}+\frac{n}{4k}}n^{\frac{3n}{8}(k+1)-\frac{k}{4}+\frac{7}{8}}e^{-\frac{n}{8}(k+1)+\frac{9n}{4}+\frac{5n}{8k}+\pi\sqrt{\frac{n}{3}(k+1)}}
\end{split}
\eeq
(2) When $k$ is even, let 
\beq
c_k = \frac{1}{N}\mathbb{E}^{(H)}_N[\Tr\,T_{k}(M_R)].
\eeq
We have,
\beq\label{eq:cheby_mu_degree_even}
\begin{split}
\mathbb{E}^{(H)}_N[X_k^n]&=\mathbb{E}^{(H)}_N\left[\left(\Tr\,T_{k}(M_R)-Nc_k\right)^n\right] \\
&= \mathbb{E}^{(H)}_N\left[\left(N\left((-1)^{\frac{k}{2}}-c_k\right) + k\sum_{j=1}^{\frac{k}{2}}(-1)^{\frac{k}{2}-j}\frac{(\frac{k}{2}+j-1)!}{(\frac{k}{2}-j)!(2j)!}2^{2j-1}\Tr\,M_R^{2j}\right)^n\right]\\
&=\sum_{n_0+\dots+n_{\frac{k}{2}}=n} \binom{n}{n_0,\dots,n_{\frac{k}{2}}}N^{n_0}\left((-1)^{\frac{k}{2}}-c_k\right)^{n_0}\\
&\qquad\times \prod_{j=1}^{\frac{k}{2}}(-1)^{\frac{k}{2}n_j-jn_j}k^{n_j}\left(\frac{(\frac{k}{2}+j-1)!}{(\frac{k}{2}-j)!(2j)!}\right)^{n_j}2^{(2j-1)n_j}\mathbb{E}^{(H)}_N\left[P_\mu\right],
\end{split}
\eeq
where
\beq\label{eq:mu for k even}
\mathbb{E}^{(H)}_N[P_\mu]=\mathbb{E}_N\left[\prod_{l=0}^{\frac{k}{2}}(\Tr\,M_R^{2l})^{n_l}\right], \quad \mu=(2^{n_1},4^{n_2},\dots,k^{n_{\frac{k}{2}}}).
\eeq
According to Szeg\H{o}'s theorem, when $n$ is even the leading order term in the R.H.S. of \eqref{eq:cheby_mu_degree_even} is given by $\mathbb{E}[(\sqrt{k}r_k/2)^n]$, $r_k\sim\mathcal{N}(0,1)$. The sub-leading term is $d_3(n,k)N^{-2}$. When $n$ is odd, the leading term in the R.H.S. is given by $d_2(n,k)N^{-1}$. Next  we compute the coefficients $d_2(n,k)$ and $d_3(n,k)$.

\noindent\textit{Coefficient $d_2(n,k)$:} 
$c_k$ decays as $1/N^2$ for $k>2$, so we neglect it in \eqref{eq:cheby_mu_degree_even}. Note that $\mu$ in \eqref{eq:mu for k even} doesn't have any 1-cycles. So $\mu$ is fixed-point-free and \eqref{eq:no fix cyc char} can also be used to estimate characters $\chi^\lambda_\mu$ in Prop.~\ref{prop:p_mu coef 1/N bound}. Here we just use \eqref{eq:p_mu coef 1/N bound} for $q=2$ and follow the exact same calculation as $k$ odd. In the limit $n \to \infty$ this leads to
\beq
d_2(n,k)\lesssim \, A^{\frac{3n}{k}}\pi^{-\frac{n}{2}}2^{\frac{3nk}{8}-3n+\frac{n}{6k}}k^{\frac{3nk}{8}+\frac{n}{2}+\frac{9n}{4k}}n^{\frac{3nk}{8}+\frac{2n}{k}-\frac{k}{2}-\frac{3}{8}}e^{-\frac{n}{8}(k-18)+\pi\sqrt{\frac{nk}{3}}-\frac{19n}{8k}.}
\eeq
Similarly, $d_{3}(n,k)$ can be approximated as
\beq
d_3(n,k)\lesssim \,A^{\frac{3n}{k}}\pi^{-\frac{n}{2}}2^{\frac{3nk}{8}-3n+\frac{n}{6k}}k^{\frac{3nk}{8}+\frac{n}{2}+\frac{9n}{4k}}n^{\frac{3nk}{8}+\frac{2n}{k}-\frac{k}{2}+\frac{5}{8}}e^{-\frac{n}{8}(k-18)+\pi\sqrt{\frac{nk}{3}}-\frac{19n}{8k}.}
\eeq
\end{proof}

Since $X_k$ converges to independent Gaussian normals, the correlators of $X_k$ also converges to random Gaussian variables as $N\rightarrow\infty$. For instance,
\beq
\mathbb{E}^{(H)}_N[X_iX_j] = \frac{\sqrt{ij}}{4}\mathbb{E}[r_ir_j]+O(N^{-1}).
\eeq
Given below are the moments corresponding to partitions of 6.
\beq
\renewcommand{\arraystretch}{1.45}
\begin{matrix*}[l]
\mathbb{E}^{(H)}_N\left[X_6\right] = 0 
&\mathbb{E}^{(H)}_N\left[X_5X_1\right] =\frac{5}{4N^2}\\
\mathbb{E}^{(H)}_N\left[X_4X_2\right] =\frac{1}{N^2} 
&\mathbb{E}^{(H)}_N\left[X_4X_1^2\right] =\frac{1}{2N}\\
\mathbb{E}^{(H)}_N\left[X_3^2\right] =\frac{3}{4} +\frac{3}{4N^2} 
&\mathbb{E}^{(H)}_N\left[X_3X_2X_1\right] =\frac{3}{4N}\\
\mathbb{E}^{(H)}_N\left[X_3X_1^3\right] =\frac{3}{8N^2}
&\mathbb{E}^{(H)}_N\left[X_2^3\right] =\frac{1}{N}\\
\mathbb{E}^{(H)}_N\left[X_2^2X_1^2\right] =\frac{1}{8} + \frac{1}{2N^2}
&\mathbb{E}^{(H)}_N\left[X_2X_1^4\right] =\frac{3}{8N} \\ 
\mathbb{E}^{(H)}_N\left[X_1^6\right]  = \frac{15}{64}
\end{matrix*}
\eeq
In this paper, we do not pursue correlations of $X_k$ any further.

\subsection{Cumulants}\label{sec:cumulants}
In general, the moments and the cumulants are related by the recurrence relation
\beq
\kappa_n = m_n - \sum_{j=1}^{n-1}\binom{n-1}{j-1}\kappa_{j}m_{n-j}.
\eeq
Cumulants and moments can also be expressed in terms of each other through a more elegant formula. Let $\lambda=(\lambda_1,\dots,\lambda_l)\equiv (1^{b_1},2^{b_2},\dots r^{b_r})$ be a partition of $n$. Define
\beq
\kappa_\lambda := \prod_{j=1}^l\kappa_j =\prod_{j=1}^r\kappa_j^{b_j},\quad m_\lambda := \prod_{j=1}^lm_j=\prod_{j=1}^rm_j^{b_j}.
\eeq
We have
\beq
\begin{split}
m_n &= \sum_\lambda d_\lambda \kappa_\lambda\\
\kappa_n &=\sum_\lambda(-1)^{l(\lambda)-1}(l(\lambda)-1)! d_\lambda m_\lambda
\end{split}
\eeq
where
\beq
d_\lambda = \frac{n!}{(1!)^{b_1}b_1!\dots (r!)^{b_r}b_r!}
\eeq
is the number of decompositions of a set of $n$ elements into disjoint subsets containing $\lambda_1,\dots,\lambda_l$ elements. 

In this section, we give an estimate on the cumulants of random variables $X_k$ and to do so we rely on the well studied connection between GUE correlators and enumerating ribbon graphs which has been briefly discussed in App.~\ref{app:ribbon graphs}.

Consider the formal matrix integral over the space of $N\times N$ rescaled GUE matrices.
\beq\label{eq:formal int}
Z_N(\textbf{s},\xi) = e^{s_0N\xi}\int e^{-2N \Tr\,M^2} e^{\xi\Tr\,V(M).}\, dM
\eeq
Here the formal series $V(M)$ depending on the parameters $\textbf{s}=\{s_0,s_1,\dots,s_k\}$ has the form
\beq
V(M) = \sum_{j=1}^ks_jM^j.
\eeq
The integral in \eqref{eq:formal int} can be considered as a formal expansion in the set of parameters $s_j$ and $\xi$.  Now,
\beq\label{eq:formal series1}
\begin{split}
\frac{Z(\textbf{s},\xi)}{Z(0,\xi)} &= \sum_{n_0,n_1,\dots ,n_k}\xi^{\sum n_j}\frac{(s_0N)^{n_0}}{n_0!}\frac{s_1^{n_1}}{n_1!}\dots\frac{s_k^{n_k}}{n_k!}\mathbb{E}^{(H)}_N\left[\prod_{j=1}^k(\Tr\, M_R^j)^{n_j}\right]\\
&=\sum_{n\geq 0} \xi^n\sum_{n_0+\dots +n_k=n }\frac{(s_0N)^{n_0}}{n_0!}\frac{s_1^{n_1}}{n_1!}\dots\frac{s_k^{n_k}}{n_k!}\mathbb{E}^{(H)}_N\left[\prod_{j=1}^k(\Tr\, M_R^j)^{n_j}\right].
\end{split}
\eeq
By choosing $s_j$ to be the coefficients of Chebyshev polynomials in \eqref{eq:formal series1}, we recover the moments of $X_k$. Thus, \eqref{eq:formal series1} is the moment generating function of $X_k$. For a given $k$, by fixing $s_j$ to be the Chebyshev coefficients in $T_k$,
\beq\label{eq:mom_gen_func}
\mathbb{E}^{(H)}_N\left[e^{\xi X_k}\right]= \sum_{n\geq 0}\frac{\xi^n}{n!}\mathbb{E}^{(H)}_N [X^n_k] = \frac{Z(\textbf{s},\xi)}{Z(0,\xi)}.
\eeq
By matching the terms in the L.H.S. and R.H.S. of \eqref{eq:mom_gen_func} by powers in $\xi$ we recover the moments of $X_k$.
The correlators of $\Tr\, M_R^j$ are connected to the problem of enumerating ribbon graphs. For a brief introduction see App.~\ref{app:ribbon graphs} and for more details see \cite{Eynard2016,Itzykson1990} and references within. The trace correlators count ribbon graphs that are connected and also multiplicatively count ribbon graphs that are disconnected. When we have a generating function that counts disconnected objects multiplicatively, taking logarithm counts only the connected objects \cite{Harer1986}. Hence, the cumulant generating function,
\beq\label{eq:cumulant expansion}
\begin{split}
\log \mathbb{E}^{(H)}_N\left[e^{\xi X_k}\right] &= \log \frac{Z(\textbf{s},\xi)}{Z(0,\xi)} =\sum_{n\geq 1}\frac{\xi^n}{n!}\kappa_n\\
&= s_0N\xi + \sum_{n\geq 1} \xi^n\sum_{n_1\dots +n_k=n }\frac{s_1^{n_1}}{n_1!}\dots\frac{s_k^{n_k}}{n_k!}\mathbb{E}^{(H)}_N\left[\prod_{j=1}^k(\Tr\, M_R^j)^{n_j}\right]_c
\end{split}
\eeq
keeps only connected ribbon graphs indicated by subscript $c$. For $\mu=(1^{n_1},\dots, k^{n_k})\equiv (\mu_1,\dots,\mu_l)$, the connected correlators are given by
\beq\label{eq:connected correlators}
\mathbb{E}^{(H)}_N\left[\prod_{j=1}^l\Tr\, M_R^{\mu_j}\right]_c\equiv \mathbb{E}^{(H)}_N\left[\prod_{j=1}^k(\Tr\, M_R^j)^{n_j}\right]_c = \sum_{0\leq g\leq \frac{|\mu|}{4}-\frac{l}{2}+\frac{1}{2}}\frac{1}{2^{|\mu|}}a_g(\mu_1,\dots,\mu_l)N^{2-2g-l},\quad |\mu|\,\,\text{is even}.
\eeq
Here 
\beq
\begin{split}
a_g(\mu_1,\dots,\mu_l) &= \#\{\text{connected oriented labelled ribbon graphs}\\\
& \qquad \text{of genus $g$ with $l$ vertices of valencies $\mu_1,\dots,\mu_l$}\}\\
&=l!\sum_\Gamma\frac{1}{\#\text{Sym}(\Gamma)},
\end{split}
\eeq
where $\Gamma$ is a connected (unlabelled) ribbon graph  of genus $g$ with $l$ vertices of valencies $\mu_1,\dots ,\mu_l$, $\#\text{Sym}(\Gamma)$ is the order of the symmetry group of $\Gamma$, and the last summation is taken over all such $\Gamma$. 
For the explicit expressions of connected correlators of the GUE see \cite{Dubrovin2017}.

We are now ready to estimate the cumulants of $X_k$. We treat $k$ even and odd cases separately.

\noindent (1) \textit{$k$ odd}: In this case, the parameters $s_{2j}=0$ for $0\leq j \leq \frac{k-1}{2}$ and 
\beq
s_{2j+1} = (-1)^{\frac{k-1}{2}-j}k\frac{(\frac{k-1}{2}+j)!}{(\frac{k-1}{2}-j)!(2j+1)!}2^{2j},\quad 0\leq j\leq \frac{k-1}{2}.
\eeq
When $k$ is odd, all the odd moments are zero. Hence all the odd cumulants are also zero. By inserting \eqref{eq:connected correlators} in \eqref{eq:cumulant expansion}, the even cumulants are given by
\beq\label{eq:cumulants k odd}
\kappa_{2n} = \frac{2n!}{N^{2n-2}}\sum_g\sum_{n_1+n_3+\dots +n_k=2n}\frac{1}{2^{\sum_jjn_j}}\frac{a_g}{N^{2g}}\frac{s_{1}^{n_{1}}}{n_{1}!}\frac{s_{3}^{n_{3}}}{n_{3}!}\dots \frac{s_{k}^{n_{k}}}{n_{k}!}.
\eeq

\noindent (2) \textit{k even:} In this case, the parameters $s_{2j+1}=0$ for $0\leq j\leq \frac{k}{2}-1$ and
\beq
\begin{split}
s_0&=(-1)^{\frac{k}{2}}-c_k,\\
s_{2j}&=(-1)^{\frac{k}{2}-j}k\frac{(\frac{k}{2}+j-1)!}{(\frac{k}{2}-j)!(2j)!}2^{2j-1},\quad 1\leq j\leq \frac{k}{2}.
\end{split}
\eeq
The first cumulant is zero by definition of $X_k$. So the first term in \eqref{eq:cumulant expansion} is cancelled by $n=1$ contribution from the second term. Hence,
\beq
\log \mathbb{E}^{(H)}_N\left[e^{\xi X_k}\right] = \sum_{n\geq 2} \xi^n\sum_{n_2\dots +n_k=n }\frac{s_2^{n_2}}{n_2!}\frac{s_4^{n_4}}{n_4!}\dots\frac{s_k^{n_k}}{n_k!}\mathbb{E}^{(H)}_N\left[\prod_{j=1}^{k/2}(\Tr\, M_R^{2j})^{n_{2j}}\right]_c
\eeq
By inserting \eqref{eq:connected correlators} in \eqref{eq:cumulant expansion}, the cumulants are given by
\beq\label{eq:cumulants k even}
\kappa_n = \frac{n!}{N^{n-2}}\sum_g\sum_{n_2+\dots +n_k=n}\frac{1}{2^{\sum_jjn_j}}\frac{a_g}{N^{2g}}\frac{s_2^{n_2}}{n_2!}\frac{s_4^{n_4}}{n_4!}\dots \frac{s_k^{n_k}}{n_k!},\quad n\geq 2.
\eeq

Third and higher order cumulants of Gaussian random variable are identically zero. Since $X_k$ converges to $\mathcal{N}(0,k/4)$ as $N\rightarrow\infty$, cumulants of $X_k$, $\kappa_n\rightarrow 0$ as $N\rightarrow\infty$ for all $n\geq 3$. For a fixed $n$, we see from \eqref{eq:cumulants k odd} and \eqref{eq:cumulants k even} that $\kappa_{n}$ decay as $N^{-n+2}$. 

\noindent\textit{Example:} The simplest non-trivial example is to calculate the cumulants of $X_2$. By mapping the problem to counting ribbon graphs (see App.~\ref{app:ribbon graphs}),
\beq
\mathbb{E}^{(H)}_N[(\Tr M^2_R)^n]_c = \frac{1}{(4N)^n}2^{n-1}(n-1)!N^2 = (n-1)!\frac{1}{2^{n+1}}\frac{1}{N^{n-2}}.
\eeq
For $X_2$, $s_0=-\frac{1}{2}$, $s_2=2$, and $s_j=0$ for $j\neq 0,2$. Hence
\beq
\kappa_n = s_2^n\,\mathbb{E}^{(H)}_N[\Tr (M_R^2)^n]_c = \frac{1}{2}\frac{(n-1)!}{N^{n-2}}.
\eeq
%

\section*{Acknowledgements}
FM is grateful for support from the  University Research Fellowship of the University of Bristol. JPK is pleased to acknowledge support from a Royal Society Wolfson Research Merit Award and ERC Advanced Grant 740900 (LogCorRM). We also thank Tamara Grava and Sergey Berezin for helpful discussions.

\section*{Data availability}
Data sharing is not applicable to this article as no new data were created or analysed in this study.

\appendix
\section*{Appendices}
\section{Some properties of multivariate orthogonal polynomials}
\renewcommand{\theequation}{A.\arabic{equation}}
\textit{Gaussian case. }For $N=1$, the multivariate Hermite polynomials coincide with the classical polynomials,
\beq\label{eq:classical_her}
H_n(x) = n!\sum_{j=0}^n\mathbbm{1}_{n-j= 0\,\text{mod}\,2}\frac{1}{\left(\frac{n-j}{2}\right)!}\frac{(-1)^{\frac{n-j}{2}}}{2^{\frac{n-j}{2}}j!}x^j,
\eeq
which have the generating function
\beq\label{eq:classical_her_genfun}
\sum_j \frac{H_j(x)}{j!}t^j = e^{xt-\frac{t^2}{2}}.
\eeq
By comparing \eqref{eq:mulher_to_schur} and \eqref{eq:classical_her}, we can see the analogies between classical Hermite polynomials and their multivariate counterparts: the sum over $j$ is replaced by the sum over partitions; the role of monomials is played by Schur polynomials; the factorials are replaced with $C_\lambda(N)$. With this analogy the generating function of $\mathcal{H}_\lambda$ is \cite{Baker1997}
\beq\label{eq:mulher_genfun_old}
\sum_{\lambda}\frac{\mathcal{H}_\lambda(\textbf{x})}{C_\lambda(N)}S_\lambda(\textbf{t}) = \left(\sum_\mu\frac{S_\mu(\textbf{x})S_\mu(\textbf{t})}{C_\mu(N)}\right)\left(\sum_{n=0}^\infty\sum_{\nu\vdash 2n}\frac{(-1)^{\frac{|\nu|}{2}}}{2^{\frac{|\nu|}{2}}}S_\nu(\textbf{t})D_{\nu 0}^{(H)}\right).
\eeq
The validity of the above formula can be easily verified for lower order partitions (say $|\lambda|=2,4$) using Pieri's formula but the second factor in \eqref{eq:mulher_genfun_old} can be simplified further.

\begin{proposition} Let $t_1,t_2,\dots$  be a set of variables, then
\beq\label{eq:my_identity}
\prod_{j}\exp\left(-\frac{t_j^2}{2}\right) = \sum_{n=0}^\infty\sum_{\nu\vdash 2n}\frac{(-1)^{\frac{|\nu|}{2}}}{2^{\frac{|\nu|}{2}}}S_\nu(\textbf{t})D_{\nu 0}^{(H)}.
\eeq
\end{proposition}
\begin{proof}
For a fixed $n$, $|\nu|=2n$. Comparing \eqref{eq:schur_to_mulher_coef} with  \eqref{eq:psi_lambda0},
\beq
D_{\nu 0}^{(H)} = \frac{1}{n!}\chi^\nu_{(2^n)}.
\eeq
Now using \eqref{eq:power_to_schur} proves the proposition.
\end{proof}
\begin{proposition}
Let $x_1,\dots,x_N$ and $t_1,\dots,t_N$ be two sets of variables. The multivariate Hermite polynomials defined in \eqref{eq:mulher_def} have the following generating function \cite{Baker1997}:
\beq\label{eq:mulher_genfun}
\sum_{\lambda}\frac{\mathcal{H}_\lambda(\textbf{x})}{C_\lambda(N)}S_\lambda(\textbf{t}) = \left(\sum_\mu\frac{S_\mu(\textbf{x})S_\mu(\textbf{t})}{C_\mu(N)}\right)\prod_{j}\exp\left(-\frac{t_j^2}{2}\right).
\eeq
\end{proposition}
Several other analogues of properties of the classical Hermite polynomials, including an integral representation, summation, integration and differentiation formulae, are given for $\beta-$ensembles in \cite{Baker1997}. Note that in \cite{Baker1997} $C^\alpha_\mu$ ($\alpha\in\mathbb{R}$) is used to denote Schur polynomials with a specific normalisation where as in this work $C_\mu(N)$ is a constant in $N$ given in \eqref{eq:C_lambda}. 

\noindent\textit{Laguerre case. }When $N=1$, $\mathcal{L}^{(\gamma)}_\lambda$ coincides with the classical Laguerre polynomials 
\beq\label{eq:classical_lag}
L_n^{(\gamma)}(x) = \sum_{j=0}^n (-1)^j \frac{\Gamma(n+\gamma +1)}{\Gamma(j+\gamma +1)(n-j)!} \frac{x^j}{j!},
\eeq
whose generating function is
\beq\label{eq:classical_lag_genfun}
\sum_{j=0}^\infty	 \frac{1}{\Gamma(j+\gamma +1)}L^{(\gamma)}_j(x)t^j = e^{t}\frac{J_\gamma(2\sqrt{tx})}{(tx)^{\frac{\gamma}{2}}} = e^t\sum_{m=0}^\infty\frac{(-1)^m}{m!\Gamma(m+\gamma +1)}(tx)^m,
\eeq
where $J_\gamma$ is the Bessel function. By comparing \eqref{eq:mullag_to_schur} and \eqref{eq:classical_lag}, the generating function for multivariate Laguerre polynomials \cite{Baker1997} is
\beq
\sum_{\nu}\frac{1}{G_\nu(N,\gamma)}\mathcal{L}^{(\gamma)}_\nu(\textbf{x})S_\nu(\textbf{t}) = (-1)^{\frac{N(N-1)}{2}}\left(\sum_\lambda S_\lambda(\textbf{t})D^{(L)}_{\lambda 0}\right)\left(\sum_\mu \frac{(-1)^{|\mu|}}{G_\mu(N,\gamma)}\frac{S_\mu(\textbf{x})S_\mu(\textbf{t})}{G_\mu(N,0)}\right),
\eeq
or equivalently using \eqref{eq:D_lambda0_lag},
\beq
\sum_{\nu}\frac{1}{G_\nu(N,\gamma)}\mathcal{L}^{(\gamma)}_\nu(\textbf{x})S_\nu(\textbf{t}) = (-1)^{\frac{N(N-1)}{2}}\left(\sum_\mu \frac{(-1)^{|\mu|}}{G_\mu(N,\gamma)}\frac{S_\mu(\textbf{x})S_\mu(\textbf{t})}{G_\mu(N,0)}\right)\prod_{j=1}^Ne^{t_j}.
\eeq

\section{Ribbon graphs and matrix integrals}\label{app:ribbon graphs}
\renewcommand{\theequation}{B.\arabic{equation}}
Let $\textbf{x} = (x_1,\dots, x_N)$ be an $N-$dimensional random variable. Consider the normalised Gaussian measure 
\beq
d\mu(\textbf{x}) = (2\pi)^{-\frac{N}{2}}\sqrt{\det A}\,e^{-\frac{1}{2}\sum_{i,j}x_iA_{ij}x_j}\prod_kdx_k, 
\eeq
where $A$ is a positive definite symmetric matrix. The inverse 
\beq
B_{ij} = (A^{-1})_{ij}
\eeq
is called the propagator.

Correlations of Gaussian random variables can be computed in a combinatorial way using \textit{Wick's theorem} \cite{Wick1950}, also known as \textit{Isserlis' theorem}, which is stated below.
\begin{theorem}[Wick's theorem] The expectation value of product of Gaussian random variables is 
\beq
\mathbb{E}[x_{i_1}x_{i_2}\dots x_{i_n}] = 
\begin{cases}
0 \hspace{16em} \text{if $n$ is odd},\\
B_{i_1i_2}\hspace{14.5em}\text{if $n=2$},\\
\sum_{\text{pairings of }(i_1,\dots, i_n)}\prod_{\text{pairs}\,(k,l)}B_{i_ki_l},\hspace{2.8em} \text{if $n\geq2$ and even}.
\end{cases}
\eeq
For example,
\beq
\mathbb{E}[x_{i_1}x_{i_2}x_{i_3}x_{i_4}] = B_{i_1i_2}B_{i_3i_4} + B_{i_1i_3}B_{i_2i_4} + B_{i_1i_4}B_{i_2i_3}.
\eeq
\end{theorem}
Wick's theorem becomes particularly useful when the indices $i_j$ are repeated. The problem of computing the expectation values $\mathbb{E}[x_{i_1}^{b_1}\dots x_{i_n}^{b_n}]$ can be mapped to counting the number of ways of gluing $n$ vertices with valencies $b_1,\dots b_n$, whose weights are determined by the propagators that correspond to their edges.
\beq\label{eq:gaussian moments graphs}
\mathbb{E}[x_{i_1}^{b_1}\dots x_{i_n}^{b_n}] = \sum_{\substack{\text{Graphs G with $n$ vertices}\\ \text{of valencies $b_j$}}}\,\,\prod_{\text{$(i_k,i_l)$ edge of G}}B_{i_ki_l}.
\eeq
For example,
\beq
\mathbb{E}[x_{i_1}^2x_{i_2}^2] = B_{i_1i_1}B_{i_2i_2} + 2B^2_{i_1i_2}.  
\eeq
Clearly many graphs in \eqref{eq:gaussian moments graphs} are topologically identical and have the same weight because of the symmetries among the edges and vertices. Let $\textbf{G}$ be the group of these symmetries, $\#gluings$ be the number of gluings of obtaining a graph, and $\text{Aut}(G)$ be the automorphism group of the graph. By orbit-stabiliser theorem,
\beq
\#\text{Aut}(G)\times \#gluings = \#\textbf{G}, 
\eeq
where $\#\textbf{G}$ is the order of group relabelling. Wick's theorem can be written only in terms of non-equivalent graphs as follows:
\beq\label{eq:wicks theorem}
\frac{1}{\#\textbf{G}}\mathbb{E}\big[\prod_jx_{i_j}^{b_j}\big] = \sum_{\text{Non-equivalent graphs $G$}}\frac{1}{\# \text{Aut}(G)}\,\prod_{(i,j)\text{ edge of $G$}}B_{ij}.
\eeq

In the case of Gaussian matrix integrals, Wick's theorem can be applied to compute correlators of traces by studying \textit{fat graphs} also called \textit{ribbon graphs}.

Consider the Hermitian Gaussian matrix model with probability measure
\beq
d\mu_0(M_R) = \frac{1}{\mathcal{Z}_0}e^{-2N\Tr M_R^2}\prod_{j=1}^NdM_{{R}_{jj}}\prod_{j<k}d\text{Re}M_{R_{jk}}\,d\text{Im}M_{R_{jk}},
\eeq
where
\beq
\mathcal{Z}_0 = \frac{1}{2^{N(N-1)}}\left(\frac{\pi}{N}\right)^{\frac{N^2}{2}}.
\eeq

The Wick's propagator is 
\beq\label{eq:propagator}
\mathbb{E}^{(H)}_N[M_{R_{ij}}M_{R_{kl}}] \equiv \br M_{R_{ij}}M_{R_{kl}}\kt= \frac{1}{4N}\delta_{il}\delta_{jk}.
\eeq
As an example, consider 
\beq
\mathbb{E}^{(H)}_N[(\Tr M_R^3)^2] = \sum_{\substack{i,j,k,\\l,m,n}}\mathbb{E}^{(H)}_N[M_{R_{ij}}M_{R_{jk}}M_{R_{ki}}M_{R_{lm}}M_{R_{mn}}M_{R_{nl}}]
\eeq
To map the problem to counting graphs, associate a vertex to each trace. The power of the matrix inside the trace gives the number of half-edges as double lines with index associated to each single line.
\begin{figure}[h]  
\centering 
\begin{tikzpicture}
\node at (-2,0.1) [left] {$\Tr  M_R^3$};
 \draw[very thick,black,->] (-1.9,0) -- (-0.8,0);
 \draw (0,1) 
  node [above left] {$i$} -- (0,0) -- (-1,-.8) node [above left] {$i$};
  \draw (-.9,-1)
  node [below right] {$k$} -- (0.1,-.2) -- (.9,-1) node [below left] {$k$};
  \draw (1,-.8)
  node [above right]{$j$} -- (0.2,0) -- (0.2,1) node [above right] {$j$};
\end{tikzpicture}
\end{figure}
The propagator in \eqref{eq:propagator} can be used to glue these half-edges together to form a double line edge of the graph.
Thus,
\beq\label{eq:trace example}
\begin{split}
\mathbb{E}^{(H)}_N[(\Tr M_R^3)^2] &= \sum_{\substack{i,j,k\\l,m,n}}\br M_{R_{ij}}M_{R_{jk}}\kt\br M_{R_{ki}}M_{R_{lm}}\kt\br M_{R_{mn}}M_{R_{nl}}\kt + \br M_{R_{ij}}M_{R_{ki}}\kt\br M_{R_{jk}}M_{R_{lm}}\kt\br M_{R_{mn}}M_{R_{nl}}\kt + \dots \\
&= \frac{1}{(4N)^3}\sum_{\substack{i,j,k\\l,m,n}} \delta_{ik}\delta_{km}\delta_{il}\delta_{ml} + \delta_{jk}\delta_{jm}\delta_{kl}\delta_{ml} + \dots\\
&=\frac{1}{4^3}\left(12+\frac{3}{N^2}\right).
\end{split}
\eeq
There are in total $5!! = 15$ graphs in \eqref{eq:trace example} with only two topologically distinct graphs shown below.
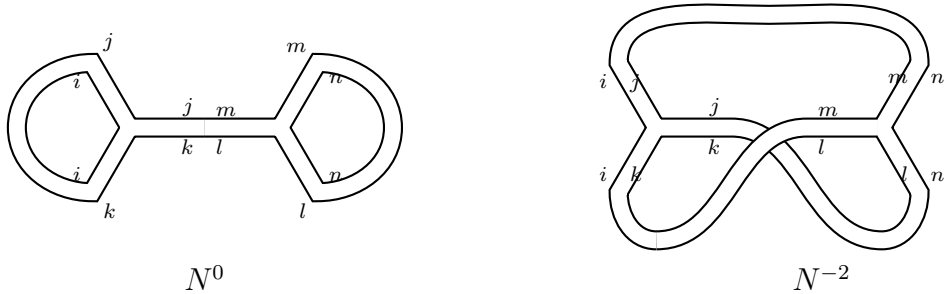
\begin{figure}[h]
\centering
\begin{tikzpicture}
  \draw [black,thick,double,double distance=6pt]
  (0,0) node [above left] {\scriptsize{$j$}} node [below left] {\scriptsize{$k$}} to (-1,0) to (-1.5,.86) node [below left] {\scriptsize{$i$}} node [above right] {\scriptsize{$j$}} to [out=left,in=up] (-2.5,0) to
  [out=down,in=left] (-1.5,-.86) node [above left] {\scriptsize{$i$}} node [below right] {\scriptsize{$k$}} to (-1,0) ;

  \draw [black,thick,double,double distance=6pt]
   (0,0) node [above right] {\scriptsize{$m$}} node [below right] {\scriptsize{$l$}} to (1,0) to (1.5,.86) node [above left] {\scriptsize{$m$}} node [below right] {\scriptsize{$n$}} to [out=right,in=up] (2.5,0) to
   [out=down,in=right] (1.5,-.86) node [above right] {\scriptsize{$n$}} node [below left] {\scriptsize{$l$}} to (1,0) ;
   
\draw [black,thick,double,double distance=6pt]
(6,0)  to (5.5,0.86) node [below left] {\scriptsize{$i$}}  node [below right] {\scriptsize{$j$}} [out=up,in=left] to  (7.5,1.5) [out=right,in=up] to  (9.5,0.86)node [below left] {\scriptsize{$m$}}  node [below right] {\scriptsize{$n$}} to [out=-120, in = 60](9.,0) to [out=-60,in=120](9.5,-0.86)node [above left] {\scriptsize{$l$}}  node [above right] {\scriptsize{$n$}} to [out=down,in=right] (9,-1.5) to [out=left,in=right](7,0) node [above left] {\scriptsize{$j$}}  node [below left] {\scriptsize{$k$}} [out=left,in=right]to (6,0) to [out=-120,in=60](5.5,-0.86)node [above left] {\scriptsize{$i$}}  node [above right] {\scriptsize{$k$}} to [out=down,in=left] (6,-1.5);
 \draw [black,thick,double,double distance=6pt] 
 (6,-1.5)  to [out=right,in=left] (8,0) node [above right] {\scriptsize{$m$}}  node [below right] {\scriptsize{$l$}} to [out=right,in=left] (8.935,0); 
  \draw (0,-2) node {$N^0$};
  \draw (8.2,-2) node {$N^{-2}$};
\end{tikzpicture}
\end{figure}

If we attach to each vertex a factor of $N$, the $N$ dependence of a graph is: There is a factor $N$ per vertex, a factor $N^{-1}$ per edge, a factor $N$ for each single line  when summed over indices. The number of single lines remaining at the end is the number of faces of the graph. So the total $N$ dependency of a graph is
\beq
N^{\# \text{vertices}-\#\text{edges}+\#\text{faces}} = N^{\chi(G)},
\eeq
where $\chi(G)$ is the topological invariant of the graph called its Euler-characteristic.

This notion of counting ribbon graphs can be extended to compute correlators of the form $\mathbb{E}^{(H)}_N[\prod_{j}(\Tr M_R^j)^{b_j}]$. When divided by $\prod_jj^{b_j}b_j!$, the order of group relabelling, matrix integrals takes a form similar to \eqref{eq:wicks theorem}. This formula is due to Brezin-Itzykson-Parisi-
Zuber in 1978 \cite{Brezin1978planar}
\beq
\mathbb{E}^{(H)}_N\left[\prod_{j=1}^n\frac{1}{b_j!}\left(\frac{N}{j}\Tr M_R^j\right)^{b_j}\right] = \sum_{\text{Ribbon Graphs $G$}}\frac{1}{\#\text{Aut($G$)}}4^{-\#\text{edges}}N^{\chi(G)},
\eeq
where the sum is over non-topologically equivalent ribbon graphs and $\#\text{Aut($G$)}$ is the number of automorphisms of $G$. There are a total of $(\sum_j jb_j-1)!!$ graphs (counting equivalent and non-equivalent graphs). The total number of vertices is $b=\sum_jb_j$ with $j$ valencies for each vertex and the total number of edges is $(\sum_jjb_j)/2$.

\subsection{Special cases}
Here we consider two cases (i) $\mathbb{E}^{(H)}_N[\Tr M_R^{2k-1}\Tr M_R]$ and (ii) $\mathbb{E}^{(H)}_N[(\Tr M_R^2)^n]$.

\noindent(i) $\mathbb{E}^{(H)}_N[\Tr M_R^{2k-1}\Tr M_R]$: We represent $\Tr M_R^{2k-1}\Tr M_R$ as two vertices with $2k-1$ and 1 valencies, respectively.
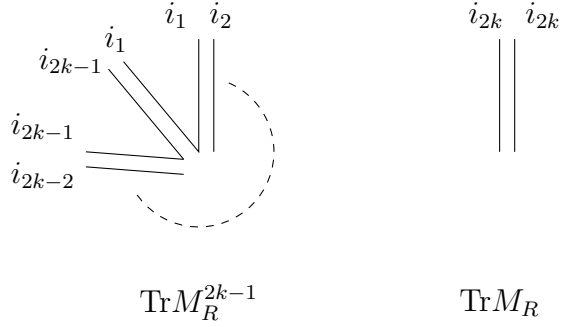
\begin{figure}[h]
\begin{center}
\begin{tikzpicture}
\draw (0,1.5) node [above left] {$i_1$} node [above right] {$i_2$}  -- (0,0) -- (-1,1.2);
\draw (-1.7,1.2) node {$i_{2k-1}$};
\draw (-1.1,1.5) node {$i_{1}$};
\draw (0.2,1.5)-- (0.2,0);
\draw (-1.2,1.1) -- (-0.2,-0.1);
\draw (-0.2,-0.1)--(-1.5,0)  node [above left] {$i_{2k-1}$} node [below left] {$i_{2k-2}$};
\draw (-0.2,-0.3)--(-1.5,-0.2);
\draw [black,dashed,domain=215:360] plot ({cos(\x)}, {sin(\x)});
\draw [black,dashed,domain=0:70] plot ({cos(\x)}, {sin(\x)});
\draw (4,1.5)--(4,0);
\draw (4.2,1.5) node [above right] {$i_{2k}$} node [above left] {$i_{2k}$}--(4.2,0) ;
\draw (0,-2) node {$\Tr M_R^{2k-1}$};
\draw (4,-2) node {$\Tr M_R$};
\end{tikzpicture}
\end{center}
\end{figure}

\noindent Since index $i_{2k}$ has $2k-1$ choices, by gluing the half-edges using \eqref{eq:propagator},
\beq
\begin{split}
\mathbb{E}^{(H)}_N[\Tr M_R^{2k-1}\Tr M_R] &= \frac{(2k-1)}{4N}\mathbb{E}^{(H)}_N[\Tr M_R^{2k-2}]\\
&=\frac{N}{k(4N)^k}(2k-1)!!i^{-k+1}P_{k-1}^{(1)}\left(iN,\frac{\pi}{2}\right),
\end{split}
\eeq 
where $P_{k-1}^{(1)}\left(iN,\frac{\pi}{2}\right)$ is a Meixner-Pollaczek polynomial.

\noindent(ii) $\mathbb{E}^{(H)}_N[(\Tr M_R^2)^n]$: Here we sketch the idea to calculate moments of $\Tr M_R^2$. We represent $(\Tr M_R^2)^n$ as $n$ vertices each with two valencies as shown below.
\begin{figure}[h]
\centering 
\begin{tikzpicture}
\draw (0,1.5) node [above right] {$i_2$} node [above left] {$i_1$}-- (0,-1.5) node [below right] {$i_2$} node [below left] {$i_1$};
\draw (0.2,1.5)  --(0.2,-1.5) ;
\draw (2,1.5) node [above right] {$i_4$} node [above left] {$i_3$}-- (2,-1.5) node [below right] {$i_4$} node [below left] {$i_3$};
\draw (2.2,1.5)--(2.2,-1.5);
\draw [black,thick,dotted] (3,0) to (4.5,0);
\draw (5,1.5) --(5,-1.5);
\draw (5.2,1.5)node [above right] {$i_{2n}$} node [above left] {$i_{2n-1}$}--(5.2,-1.5)node [below right] {$i_{2n}$} node [below left] {$i_{2n-1}$};
\draw (0,-3) node {$\Tr M_R^2$};
\draw (2,-3) node {$\Tr M_R^2$};
\draw (5,-3) node {$\Tr M_R^2$};
\end{tikzpicture}
\end{figure}
There are several ways of gluing this set of vertices and half-edges. Trivially $i_j$ can be glued with itself for $j=1,\dots,2n$ which  gives a total contribution of $N^{2n}/(4N)^n$. 

The next non-trivial contribution comes from choosing any two vertices and gluing their valencies to form an edge between them. There are $\binom{n}{2}$ ways of choosing two vertices. Let $(i_{p},i_{p+1})$ and $(i_{q},i_{q+1})$, $1\leq p,q\leq 2n$, be the indices of the valencies of these two vertices. There are two ways to pair $(i_{p},i_{p+1})$ and $(i_{q},i_{q+1})$. This gives a contribution of $n(n-1)N^{2}/(4N)^2$. The remaining $n-2$ disconnected graphs multiplicatively gives $N^{2n-4}/(4N)^{n-2}$. Hence the first two leading terms are
\beq
\mathbb{E}^{(H)}_N[(\Tr M_R^2)^n] = \frac{1}{(4N)^n}(N^{2n} + n(n-1)N^{2n-2} +\dots)
\eeq

Remaining terms in the $n^{th}$ moment can be likewise computed. 
\beq
\mathbb{E}^{(H)}_N[(\Tr M_R^2)^n] = \frac{1}{(4N)^n}\prod_{j=0}^{n-1}(N^2+2j).
\eeq 

Similar arguments can be used to show that
\beq
\mathbb{E}^{(H)}_N[(\Tr M_R^2)^k(\Tr M_R)^{2n-2k}] = (2n-2k-1)!!\frac{1}{(4N)^n}N^{n-k}\prod_{l=n-k}^{n-1} (N^2+2l)
\eeq 
for $k=1,\dots, n-1$.
 \bibliographystyle{abbrv}
\bibliography{mom_clt}{}
\end{document}